\documentclass[final,3p]{elsarticle}
 \usepackage{graphics}
 \usepackage{graphicx}
 \usepackage{epsfig}
\usepackage{amssymb}
 \usepackage{amsthm}
 \usepackage{lineno}
 \usepackage{amsmath}
   \numberwithin{equation}{section}
\usepackage{mathrsfs}

\newtheorem{thm}{Theorem}[section]

\newtheorem{lem}[thm]{Lemma}
\newtheorem{prop}[thm]{Proposition}
\newtheorem{defn}[thm]{Definition}

\biboptions{sort&compress,square}
\begin{document}
\begin{frontmatter}
\author[rvt1,rvt2]{Jian Wang}
\ead{wangj068@gmail.com}
\author[rvt1]{Yong Wang\corref{cor2}}
\ead{wangy581@nenu.edu.cn}
\cortext[cor2]{Corresponding author.}

\address[rvt1]{School of Mathematics and Statistics, Northeast Normal University,
Changchun, 130024, P.R.China}
\address[rvt2]{Chengde Petroleum College,
Chengde, Hebei, 067000, P.R.China}

\title{A Kastler-Kalau-Walze Type Theorem for $5$-dimensional \\ Manifolds with  Boundary}
\begin{abstract}
The Kastler-Kalau-Walze theorem, announced by Alain Connes, shows that the Wodzicki residue of the inverse square of the Dirac operator
 is proportional to the Einstein-Hilbert action of general relativity. In this paper, we prove a  Kastler-Kalau-Walze type theorem
for $5$-dimensional manifolds with  boundary.
\end{abstract}
\end{frontmatter}
\section{Introduction}
\label{1}
The noncommutative residue found in \cite{Gu,Wo} plays a prominent role in noncommutative geometry.
For one-dimensional manifolds, the noncommutative residue was discovered by Adler \cite{MA}
 in connection with geometric aspects of nonlinear partial differential equations.
 For arbitrary closed compact $n$-dimensional manifolds, the noncommutative reside was introduced by Wodzicki in \cite{Wo} using the theory of zeta
 functions of elliptic pseudodifferential operators.
In \cite{Co1}, Connes used the noncommutative residue to derive a conformal 4-dimensional Polyakov action analogy.
Furthermore, Connes made a challenging observation that the noncommutative residue of the square of the inverse of the
Dirac operator was proportional to the Einstein-Hilbert action in \cite{Co2}.
Let $s$ be the scalar curvature and Wres denote  the noncommutative residue. Then the Kastler-Kalau-Walze
theorem gives an operator-theoretic explanation of the gravitational action and says that for a $4-$dimensional closed spin manifold,
 there exists a constant $c_0$, such that
 \begin{equation*}
{\rm  Wres}(D^{-2})=c_0\int_Ms{\rm dvol}_M.
\end{equation*}
In \cite{Ka}, Kastelr gave a brute-force proof of this theorem. In \cite{KW}, Kalau and Walze proved this theorem in the
normal coordinates system simultaneously. And then, Ackermann proved that
the Wodzicki residue ${\rm  Wres}(D^{-2})$ in turn is essentially the second  coefficient
of the heat kernel expansion of $D^{2}$ in \cite{Ac}.

On the other hand, Fedosov etc. defined a noncommutative residue on Boutet de Monvel's algebra and proved that it was a
unique continuous trace in \cite{FGLS}. In \cite{S}, Schrohe gave the relation between the Dixmier trace and the noncommutative residue for
manifolds with boundary.  For an oriented spin manifold  $M$ with boundary $\partial M$,  by the composition formula in Boutet de Monvel's algebra and
the definition of $ \widetilde{{\rm Wres}}$ \cite{Wa1},  $\widetilde{{\rm Wres}}[(\pi^+D^{-1})^2]$ should be the sum of two terms from
interior and boundary of $M$, where $\pi^+D^{-1}$ is an element in Boutet de Monvel's algebra  \cite{Wa1}.
 It is well known that the  gravitational action for manifolds with boundary is
also the sum of two terms from interior and boundary of $M$ \cite{H}. Considering the Kastler-Kalau-Walze Theorem for manifolds without boundary,
then the term from interior is proportional to gravitational action from interior, so it is natural to hope to get the gravitational action
for manifolds with boundary by computing $\widetilde{{\rm Wres}}[(\pi^+D^{-1})^2]$.
Based on the motivation, Wang \cite{Wa3} proved a Kastler-Kalau-Walze type theorem for $4$-dimensional  spin manifolds with boundary
  \begin{equation*}
 \widetilde{{\rm Wres}}[(\pi^+D^{-1})^2]=-\frac{\Omega_3}{3}\int_Ms{\rm dvol}_M,
\end{equation*}
where $\Omega_3$ is the canonical volume of $S^{3}$. Furthermore, Wang \cite{Wa4} found a Kastler-Kalau-Walze type theorem for higher dimensional manifolds with boundary and generalized the definition
of lower  dimensional volumes in \cite{RP} to  manifolds with boundary. For $5$-dimensional  spin  manifolds with boundary  \cite{Wa4}, Wang get
  \begin{equation*}
\widetilde{{\rm Wres}}[(\pi^+D^{-2})^2]=\frac{\pi i}{2}\Omega_2{\rm vol}_{\partial M},
\end{equation*}
 and for $6$-dimensional spin manifolds with boundary
  \begin{equation*}
\widetilde{{\rm Wres}}[(\pi^+D^{-2})^2]=-\frac{5\Omega_5}{3}\int_Ms{\rm dvol}_M.
\end{equation*}
In order to get the  boundary term,  we computed the lower dimensional volume ${\rm Vol}^{(1,3)}_6$ for $6$-dimensional spin manifolds with
boundary associated with $D^{-1}$, $D^{-3}$ in \cite{WW2}, and obtained the volume with the  boundary term
\begin{equation*}
\widetilde{{\rm Wres}}[\pi^+D^{-1}\circ\pi^+D^{-3}]=-\frac{5\Omega_4}{3}\int_Ms{\rm dvol}_M
 +\pi \Omega_3\int_{\partial M}Kd{\rm vol}_{\partial M},
\end{equation*}
where $K$ is the extrinsic curvature.

 Recently, Ackermann and Tolksdorf \cite{AT} proved a generalized version of the well-known Lichnerowicz formula for the square of the
most general Dirac operator with torsion  $D_{T}$ on an even-dimensional spin manifold associated to a metric connection with torsion.
 Meanwhile,  Pf$\ddot{a}$ffle and Stephan considered compact Riemannian spin manifolds without boundary equipped with orthogonal connections,
and investigated the induced Dirac operators in \cite{PS}. In \cite{PS1}, Pf$\ddot{a}$ffle and Stephan considered
orthogonal connections with arbitrary torsion on compact Riemannian manifolds, and for the induced
Dirac operators, twisted Dirac operators and Dirac operators of Chamseddine-Connes type they computed the spectral
action. For the associated Dirac operators with torsion $D^{*}_{T}, D_{T}$ \cite{WWY},
we got the Kastler-Kalau-Walze theorem associated to Dirac operators with
torsion on $4$-dimensional compact manifolds with boundary
 \begin{equation*}
\widetilde{Wres}[\pi^{+}(D_{T}^{*})^{-1} \circ\pi^{+}D_{T}^{-1}]=-\frac{1}{48\pi^{2}}\int_{M}\tilde{R}(x)\texttt{d}x
-\int_{\partial_{M}}\sum_{i}A_{iin}\pi \Omega_{2}\texttt{d}x',
\end{equation*}
where  definitions of $\tilde{R}(x)$, $A_{iin}$, see  \cite{PS1}.
In addition,  we proved the Kastler-Kalau-Walze type theorems for foliations with or without boundary
associated with sub-Dirac operators  in \cite{WW1}
\begin{equation*}
 \widetilde{{\rm Wres}}[(\pi^+D_{F}^{-1})^2]=-\frac{1}{24\sqrt{2}\cdot2^p\pi^{p+\frac{q}{2}+1}}\int_Ms_M{\rm dvol}_M.
\end{equation*}

In fact, in previous papers, we computed  $\widetilde{{\rm Wres}}[\pi^+D^{-p_{1}}\circ \pi^+D^{-p_{2}}]$ for $n$-dimensional spin manifolds with boundary
in case of
$n-p_{1}-p_{2}\leq 2$. In the present paper, we shall restrict our attention to the case of $n-p_{1}-p_{2}= 3$. We compute
$\widetilde{{\rm Wres}}[(\pi^+D^{-1})^2]$ for $5$-dimensional manifolds with boundary.
Our main result is as follows.

{\bf Main Theorem:}  The following identity holds
  \begin{equation*}
 \widetilde{{\rm Wres}}[(\pi^+D^{-1})^2]=\frac{\pi^{3}}{16}\int_{\partial_{M}}\Big( \frac{225}{64}K^{2}
 +\frac{29}{4}s_{M}\big|_{\partial_{M}}
 + \big(\frac{197}{12}+3i \big)s_{\partial_{M}}\Big){\rm dvol}_{\partial_{M}}.
\end{equation*}
where $s_{M}$, $s_{\partial_{M}}$ are respectively scalar curvatures on $M$ and $\partial_{M}$.
Compared with the previous results, up to the extrinsic curvature, the scalar curvature on $\partial_{M}$ and  the scalar curvature on $M$
appear in the boundary term.  This case essentially makes the whole calculations more difficult, and the boundary term  is the sum of
fifteen terms. As in computations of the boundary term, we  shall consider some new traces of multiplication
of Clifford elements. And the inverse 3-order symbol of the Dirac operator and higher derivatives of  -1-order, -2-order symbols of the Dirac operators
will be extensively used.

 This paper is organized as follows: In Section 2, we define lower dimensional volumes of compact Riemannian manifolds
 with  boundary. In Section 3, for $5$-dimensional spin manifolds with boundary and  the associated Dirac operators,
 we compute $ \widetilde{{\rm Wres}}[(\pi^+D^{-1})^2]$ and get a Kastler-Kalau-Walze type theorem in this case.

\section{Lower-Dimensional Volumes of Spin Manifolds  with  boundary}

 In this section we consider an $n$-dimensional oriented Riemannian manifold $(M, g^{M})$ with boundary $\partial_{ M}$ equipped
with a fixed spin structure. We assume that the metric $g^{M}$ on $M$ has
the following form near the boundary
 \begin{equation}
 g^{M}=\frac{1}{h(x_{n})}g^{\partial M}+\texttt{d}x _{n}^{2} ,
\end{equation}
where $g^{\partial M}$ is the metric on $\partial M$. Let $U\subset
M$ be a collar neighborhood of $\partial M$ which is diffeomorphic $\partial M\times [0,1)$. By the definition of $h(x_n)\in C^{\infty}([0,1))$
and $h(x_n)>0$, there exists $\tilde{h}\in C^{\infty}((-\varepsilon,1))$ such that $\tilde{h}|_{[0,1)}=h$ and $\tilde{h}>0$ for some
sufficiently small $\varepsilon>0$. Then there exists a metric $\hat{g}$ on $\hat{M}=M\bigcup_{\partial M}\partial M\times
(-\varepsilon,0]$ which has the form on $U\bigcup_{\partial M}\partial M\times (-\varepsilon,0 ]$
 \begin{equation}
\hat{g}=\frac{1}{\tilde{h}(x_{n})}g^{\partial M}+\texttt{d}x _{n}^{2} ,
\end{equation}
such that $\hat{g}|_{M}=g$.
We fix a metric $\hat{g}$ on the $\hat{M}$ such that $\hat{g}|_{M}=g$.

Let us give the expression of Dirac operators near the boundary. Set  $\widetilde{E}_{n}=\frac{\partial}{\partial x_{n}}$,
 $\widetilde{E}_{j}=\sqrt{h(x_{n})}E_{j}~~(1\leq j \leq n-1)$, where  $\{E_{1},\cdots,E_{n-1}\}$ are orthonormal basis of $T\partial_{M}$.
  Let $\nabla^L$ denote the Levi-civita connection
about $g^M$.
 In the local coordinates $\{x_i; 1\leq i\leq n\}$ and the fixed orthonormal frame $\{\widetilde{E}_{1},\cdots,\widetilde{E}_{n}\}$,
 the connection matrix
 $(\omega_{s,t})$
is defined by
 \begin{equation}
 \nabla^L(\widetilde{E}_{1},\cdots,\widetilde{E}_{n})^{t}= (\omega_{s,t})(\widetilde{E}_{1},\cdots,\widetilde{E}_{n})^{t}.
 \end{equation}
 The Dirac operator is defined by
 \begin{equation}
D=\sum^n_{j=1}c(\widetilde{E_{j}})\Big[\widetilde{E_j}+\frac{1}{4}\sum_{s,t}\omega_{s,t}(\widetilde{E_j})c(\widetilde{E_s})c(\widetilde{E_t})\Big].
\end{equation}
By Lemma 6.1 in \cite{WW1} and Proposition 2.2, Proposition 2.4 in \cite{DB}, we have
\begin{lem}
Let $f=\frac{1}{\sqrt{h}} $ and $\tilde{M}=I\times _{f}M$ be a Riemannian manifold with the metric $g_{f}=\texttt{d}x_{n}^{2}+f^{2}(x_{n})g$.
For vector fields $X,Y$ in $\mathcal{L}(M)$, then
\begin{eqnarray}
&&(1) \ \tilde{\nabla}_{\partial_{x_{n}}}\partial_{x_{n}}=0;\\
&&(2) \ \tilde{\nabla}_{\partial_{x_{n}}} X=\tilde{\nabla}_{X}\partial_{x_{n}}=(ln f)'X;\\
&&(3) \ \nabla_{X}Y=\nabla_{X}^{M}Y-\frac{g(X,Y)}{f}\texttt{grad} (f).
\end{eqnarray}
\end{lem}
Denote $A_{j s}^{t}=2\langle \nabla^{ L, \partial_{ M}}_{E_{j}} E_{s}, E_{t}\rangle$, then we obtain
\begin{lem}
 The following identity holds:
\begin{eqnarray}
&&(1) \ \langle \nabla^{ L}_{ \widetilde{E}_{i}}\partial_{x_{n}}, \widetilde{E}_{j}\rangle=-\frac{h'}{2h};\\
&&(2) \ \langle \nabla^{ L}_{ \widetilde{E}_{i}} \widetilde{E}_{j}, \partial_{x_{n}}\rangle=\frac{h'}{2h};\\
&&(3) \ \langle \nabla^{ L}_{ \widetilde{E}_{j}} \widetilde{E}_{s}, \widetilde{E}_{t}\rangle=\frac{\sqrt{h}}{2} A_{j s}^{t}.
\end{eqnarray}
Others are zeros.
\end{lem}
By Lemma 2.2, we have
\begin{defn}
 The following identity holds in the coordinates near the boundary
\begin{equation}
D=\sum^n_{\beta=1}c(\widetilde{E_\beta})\widetilde{E_\beta}-\frac{h'}{h}c(\texttt{d}x_{n})
+\frac{\sqrt{h}}{8}\sum_{s,\alpha}A_{\beta s}^{\alpha}
c(\widetilde{E_\beta })c(\widetilde{E_s})c(\widetilde{E_{\alpha}}).
\end{equation}
\end{defn}

To define the lower dimensional volume, some basic facts and formulae about Boutet de Monvel's calculus which can be found  in Sec.2 in \cite{Wa1}
are needed.

 Let  \begin{equation}
  F:L^2({\bf R}_t)\rightarrow L^2({\bf R}_v);~F(u)(v)=\int e^{-ivt}u(t)dt
  \end{equation}
   denote the Fourier transformation and
$\Phi(\overline{{\bf R}^+}) =r^+\Phi({\bf R})$ (similarly define $\Phi(\overline{{\bf R}^-}$)), where $\Phi({\bf R})$
denotes the Schwartz space and
  \begin{equation}
r^{+}:C^\infty ({\bf R})\rightarrow C^\infty (\overline{{\bf R}^+});~ f\rightarrow f|\overline{{\bf R}^+};~
 \overline{{\bf R}^+}=\{x\geq0;x\in {\bf R}\}.
\end{equation}
We define $H^+=F(\Phi(\overline{{\bf R}^+}));~ H^-_0=F(\Phi(\overline{{\bf R}^-}))$ which are orthogonal to each other. We have the following
 property: $h\in H^+~(H^-_0)$ iff $h\in C^\infty({\bf R})$ which has an analytic extension to the lower (upper) complex
half-plane $\{{\rm Im}\xi<0\}~(\{{\rm Im}\xi>0\})$ such that for all nonnegative integer $l$,
 \begin{equation}
\frac{d^{l}h}{d\xi^l}(\xi)\sim\sum^{\infty}_{k=1}\frac{d^l}{d\xi^l}(\frac{c_k}{\xi^k})
\end{equation}
as $|\xi|\rightarrow +\infty,{\rm Im}\xi\leq0~({\rm Im}\xi\geq0)$.

 Let $H'$ be the space of all polynomials and $H^-=H^-_0\bigoplus H';~H=H^+\bigoplus H^-.$ Denote by $\pi^+~(\pi^-)$ respectively the
 projection on $H^+~(H^-)$. For calculations, we take $H=\widetilde H=\{$rational functions having no poles on the real axis$\}$ ($\tilde{H}$
 is a dense set in the topology of $H$). Then on $\tilde{H}$,
 \begin{equation}
\pi^+h(\xi_0)=\frac{1}{2\pi i}\lim_{u\rightarrow 0^{-}}\int_{\Gamma^+}\frac{h(\xi)}{\xi_0+iu-\xi}d\xi,
\end{equation}
where $\Gamma^+$ is a Jordan close curve included ${\rm Im}\xi>0$ surrounding all the singularities of $h$ in the upper half-plane and
$\xi_0\in {\bf R}$. Similarly, define $\pi^{'}$ on $\tilde{H}$,
 \begin{equation}
\pi'h=\frac{1}{2\pi}\int_{\Gamma^+}h(\xi)d\xi.
\end{equation}
So, $\pi'(H^-)=0$. For $h\in H\bigcap L^1(R)$, $\pi'h=\frac{1}{2\pi}\int_{R}h(v)dv$ and for $h\in H^+\bigcap L^1(R)$, $\pi'h=0$.

Let $M$ be an $n$-dimensional compact oriented manifold with boundary $\partial M$.
Denote by $\mathcal{B}$ Boutet de Monvel's algebra, we recall the main theorem in \cite{FGLS}.
\begin{thm}\label{th:32}{\bf(Fedosov-Golse-Leichtnam-Schrohe)}
 Let $X$ and $\partial X$ be connected, ${\rm dim}X=n\geq3$,
 $A=\left(\begin{array}{lcr}\pi^+P+G &   K \\
T &  S    \end{array}\right)$ $\in \mathcal{B}$ , and denote by $p$, $b$ and $s$ the local symbols of $P,G$ and $S$ respectively.
 Define:
 \begin{eqnarray}
{\rm{\widetilde{Wres}}}(A)&=&\int_X\int_{\bf S}{\rm{tr}}_E\left[p_{-n}(x,\xi)\right]\sigma(\xi)dx \nonumber\\
&&+2\pi\int_ {\partial X}\int_{\bf S'}\left\{{\rm tr}_E\left[({\rm{tr}}b_{-n})(x',\xi')\right]+{\rm{tr}}
_F\left[s_{1-n}(x',\xi')\right]\right\}\sigma(\xi')dx',
\end{eqnarray}
Then~~ a) ${\rm \widetilde{Wres}}([A,B])=0 $, for any
$A,B\in\mathcal{B}$;~~ b) It is a unique continuous trace on
$\mathcal{B}/\mathcal{B}^{-\infty}$.
\end{thm}
 Let $p_{1},p_{2}$ be nonnegative integers and $p_{1}+p_{2}\leq n$. Then by Sec 2.1 of \cite{Wa3},  we have
\begin{defn} Lower-dimensional volumes of spin manifolds with boundary  are defined by
   \begin{equation}\label{}
  {\rm Vol}^{(p_1,p_2)}_nM:=\widetilde{{\rm Wres}}[\pi^+D^{-p_1}\circ\pi^+D^{-p_2}].
\end{equation}
\end{defn}

 Denote by $\sigma_{l}(A)$ the $l$-order symbol of an operator A. An application of (2.1.4) in \cite{Wa1} shows that
\begin{equation}
\widetilde{{\rm Wres}}[\pi^+D^{-p_1}\circ\pi^+D^{-p_2}]=\int_M\int_{|\xi|=1}{\rm
trace}_{S(TM)}[\sigma_{-n}(D^{-p_1-p_2})]\sigma(\xi)\texttt{d}x+\int_{\partial
M}\Phi,
\end{equation}
where
 \begin{eqnarray}
\Phi&=&\int_{|\xi'|=1}\int^{+\infty}_{-\infty}\sum^{\infty}_{j, k=0}
\sum\frac{(-i)^{|\alpha|+j+k+1}}{\alpha!(j+k+1)!}
 {\rm trace}_{S(TM)}
\Big[\partial^j_{x_n}\partial^\alpha_{\xi'}\partial^k_{\xi_n}
\sigma^+_{r}(D^{-p_1})(x',0,\xi',\xi_n)\nonumber\\
&&\times\partial^\alpha_{x'}\partial^{j+1}_{\xi_n}\partial^k_{x_n}\sigma_{l}
(D^{-p_2})(x',0,\xi',\xi_n)\Big]d\xi_n\sigma(\xi')\texttt{d}x',
\end{eqnarray}
and the sum is taken over $r-k+|\alpha|+\ell-j-1=-n,r\leq-p_{1},\ell\leq-p_{2}$.

 \section{ A Kastler-Kalau-Walze type theorem for $5$-dimensional spin manifolds with boundary }
In this section, we compute the lower dimensional volume for 5-dimensional compact manifolds with boundary and get a
Kastler-Kalau-Walze type formula in this case.
From now on we always assume that $M$ carries a spin structure so that the spinor bundle and the
Dirac operator are defined on $M$.

The following proposition is the key of the computation of lower-dimensional volumes of spin
manifolds with boundary.
\begin{prop}\cite{Wa4}  The following identity holds:
 \begin{eqnarray}
&& 1)~  When~ p_1+p_2=n,~ then, ~ {\rm Vol}^{(p_1,p_2)}_nM=c_0{\rm Vol}_M;\\
&& 2)~  when~  p_1+p_2\equiv n ~ {\rm mod}~  1, ~{\rm Vol}^{(p_1,p_2)}_nM=\int_{\partial M}\Phi.
\end{eqnarray}
\end{prop}
Nextly, for $5$-dimensional spin manifolds with boundary, we compute
${\rm Vol}^{(1,1)}_5$. By Proposition 3.1, we have
 \begin{equation}
 \widetilde{{\rm Wres}}[(\pi^+D^{-1})^{2}]=\int_{\partial M}\Phi.
\end{equation}

 Recall  the Dirac operator $D$ of the definition 2.3.
 Write
  \begin{equation}
D_x^{\alpha}=(-\sqrt{-1})^{|\alpha|}\partial_x^{\alpha};~\sigma(D)=p_1+p_0;
~\sigma(D^{-1})=\sum^{\infty}_{j=1}q_{-j}.
\end{equation}
By the composition formula of psudodifferential operators, then we have
\begin{eqnarray}
1=\sigma(D\circ D^{-1})&=&\sum_{\alpha}\frac{1}{\alpha!}\partial^{\alpha}_{\xi}[\sigma(D)]D^{\alpha}_{x}[\sigma(D^{-1})]\nonumber\\
&=&(p_1+p_0)(q_{-1}+q_{-2}+q_{-3}+\cdots)\nonumber\\
& &~~~+\sum_j(\partial_{\xi_j}p_1+\partial_{\xi_j}p_0)(
D_{x_j}q_{-1}+D_{x_j}q_{-2}+D_{x_j}q_{-3}+\cdots)\nonumber\\
&=&p_1q_{-1}+(p_1q_{-2}+p_0q_{-1}+\sum_j\partial_{\xi_j}p_1D_{x_j}q_{-1})\nonumber\\
   &&~~~+(p_1q_{-3}+p_0q_{-2}+\sum_j\partial_{\xi_j}p_1D_{x_j}q_{-2}) +\cdots .
\end{eqnarray}
Thus, we get
\begin{eqnarray}
q_{-1}&=&p_1^{-1};  \\
q_{-2}&=&-p_1^{-1}\Big[p_0p_1^{-1}+\sum_j\partial_{\xi_j}p_1D_{x_j}(p_1^{-1})\Big]; \\
q_{-3}&=&-p_1^{-1}\Big[p_0q_{-2}+\sum_{j=1}^{n-1}c(\texttt{d}x_{j})\partial_{x_j}q_{-2} +c(\texttt{d}x_{n})\partial_{x_n}q_{-2}\Big].
\end{eqnarray}

 By Lemma 2.1 in \cite{Wa3}, we have
 \begin{lem}\label{le:31}
The symbol of the Dirac operator
\begin{eqnarray}
\sigma_{-1}(D^{-1})&=&\frac{\sqrt{-1}c(\xi)}{|\xi|^{2}}; \\
\sigma_{-2}(D^{-1})&=&\frac{c(\xi)p_{0}c(\xi)}{|\xi|^{4}}+\frac{c(\xi)}{|\xi|^{6}}\sum_{j}c(\texttt{d}x_{j})
\Big[\partial_{x_{j}}(c(\xi))|\xi|^{2}-c(\xi)\partial_{x_{j}}(|\xi|^{2})\Big],
\end{eqnarray}
where
 \begin{equation}
p_{0}=-\frac{h'}{h}c(\texttt{d}x_{n})
+\frac{\sqrt{h}}{8}\sum_{s,\alpha}A_{\beta s}^{\alpha}
c(\widetilde{E_\beta })c(\widetilde{E_s})c(\widetilde{E_{\alpha}}).
\end{equation}
\end{lem}

Since $\Phi$ is a global form on $\partial M$, so for any fixed point $x_{0}\in\partial M$, we can choose the normal coordinates
$U$ of $x_{0}$ in $\partial M$(not in $M$) and compute $\Phi(x_{0})$ in the coordinates $\widetilde{U}=U\times [0,1)$ and the metric
$\frac{1}{h(x_{n})}g^{\partial M}+\texttt{d}x _{n}^{2}$. The dual metric of $g^{M}$ on $\widetilde{U}$ is
$h(x_{n})g^{\partial M}+\texttt{d}x _{n}^{2}.$ Write
$g_{ij}^{M}=g^{M}(\frac{\partial}{\partial x_{i}},\frac{\partial}{\partial x_{j}})$;
$g^{ij}_{M}=g^{M}(d x_{i},dx_{j})$, then

\begin{equation}
[g_{i,j}^{M}]=
\begin{bmatrix}\frac{1}{h( x_{n})}[g_{i,j}^{\partial M}]&0\\0&1\end{bmatrix};\quad
[g^{i,j}_{M}]=\begin{bmatrix} h( x_{n})[g^{i,j}_{\partial M}]&0\\0&1\end{bmatrix},
\end{equation}
and
\begin{equation}
\partial_{x_{s}} g_{ij}^{\partial M}(x_{0})=0,\quad 1\leq i,j\leq n-1;\quad g_{i,j}^{M}(x_{0})=\delta_{ij}.
\end{equation}

Let $\{E_{1},\cdots, E_{n-1}\}$ be an orthonormal frame field in $U$ about $g^{\partial M}$ which is parallel along geodesics and
$E_{i}=\frac{\partial}{\partial x_{i}}(x_{0})$, then $\{\widetilde{E_{1}}=\sqrt{h(x_{n})}E_{1}, \cdots,
\widetilde{E_{n-1}}=\sqrt{h(x_{n})}E_{n-1},\widetilde{E_{n}}=dx_{n}\}$ is the orthonormal frame field in $\widetilde{U}$ about $g^{M}.$
Locally $S(TM)|\widetilde{U}\cong \widetilde{U}\times\wedge^{*}_{C}(\frac{n}{2}).$ Let $\{f_{1},\cdots,f_{n}\}$ be the orthonormal basis of
$\wedge^{*}_{C}(\frac{n}{2})$. Take a spin frame field $\sigma: \widetilde{U}\rightarrow Spin(M)$ such that
$\pi\sigma=\{\widetilde{E_{1}},\cdots, \widetilde{E_{n}}\}$ where $\pi: Spin(M)\rightarrow O(M)$ is a double covering, then
$\{[\sigma, f_{i}], 1\leq i\leq 4\}$ is an orthonormal frame of $S(TM)|_{\widetilde{U}}.$ In the following, since the global form $\Phi$
is independent of the choice of the local frame, so we can compute $\texttt{tr}_{S(TM)}$ in the frame $\{[\sigma, f_{i}], 1\leq i\leq 4\}$.
Let $\{\hat{E}_{1},\cdots,\hat{E}_{n}\}$ be the canonical basis of $R^{n}$ and
$c(\hat{E}_{i})\in cl_{C}(n)\cong Hom(\wedge^{*}_{C}(\frac{n}{2}),\wedge^{*}_{C}(\frac{n}{2}))$ be the Clifford action. By \cite{Y}, then

\begin{equation}
c(\widetilde{E_{i}})=[(\sigma,c(\hat{E}_{i}))]; \quad c(\widetilde{E_{i}})[(\sigma, f_{i})]=[\sigma,(c(\hat{E}_{i}))f_{i}]; \quad
\frac{\partial}{\partial x_{i}}=[(\sigma,\frac{\partial}{\partial x_{i}})],
\end{equation}
then we have $\frac{\partial}{\partial x_{i}}c(\widetilde{E_{i}})=0$ in the above frame. By Lemma 2.2 in \cite{Wa3}, we have

\begin{lem}\label{le:32}
With the metric $g^{M}$ on $M$ near the boundary
\begin{eqnarray}
\partial_{x_j}(|\xi|_{g^M}^2)(x_0)&=&\left\{
       \begin{array}{c}
        0,  ~~~~~~~~~~ ~~~~~~~~~~ ~~~~~~~~~~~~~{\rm if }~j<n; \\[2pt]
       h'(0)|\xi'|^{2}_{g^{\partial M}},~~~~~~~~~~~~~~~~~~~~~{\rm if }~j=n.
       \end{array}
    \right. \\
\partial_{x_j}[c(\xi)](x_0)&=&\left\{
       \begin{array}{c}
      0,  ~~~~~~~~~~ ~~~~~~~~~~ ~~~~~~~~~~~~~{\rm if }~j<n;\\[2pt]
\partial x_{n}(c(\xi'))(x_{0}), ~~~~~~~~~~~~~~~~~{\rm if }~j=n,
       \end{array}
    \right.
\end{eqnarray}
where $\xi=\xi'+\xi_{n}\texttt{d}x_{n}$.
\end{lem}
Then
\begin{lem}\label{le:32}
With the metric $g^{M}$ on $M$ near the boundary
\begin{eqnarray}
\partial_{x_i}\partial_{x_j}(|\xi|_{g^M}^2)(x_0)\Big|_{|\xi'|=1}&=&\left\{
       \begin{array}{c}
        ~~~~~~0,  ~~~~~~~~~~  ~~~~~~~~~~~~~~~~{\rm if }~i<n,j=n;~ or~ i=n,j<n; \\  \\
          -\frac{1}{3} \sum_{\alpha ,\beta <n}\Big(R^{\partial_{M}}_{i\alpha j\beta}(x_0)
          +R^{\partial_{M}}_{i\beta j\alpha}(x_0)\Big)\xi_{\alpha}\xi_{\beta},
           ~~ ~~~~~~~{\rm if }~i,j<n;\\[2pt]   \\
h''(0), ~~~~~~~~~~~~~~~~~~~~~~~~~~~~~~~~~~~~~~~~~~~~~~~~{\rm if }~i=j=n.
       \end{array}
    \right. \\  \nonumber\\
\partial_{x_i}\partial_{x_j}[c(\xi)](x_0)\Big|_{|\xi'|=1}&=&\left\{
       \begin{array}{c}
      ~~~~~~~~~0,  ~~~~~~~~~~ ~~~~~~~~~~~~~{\rm if }~i<n,j=n;~ or~ i=n,j<n; \\  \\
       \frac{1}{6} \sum_{l,t<n}\xi_{l}\Big(R^{\partial_{M}}_{tilj}(x_0)+R^{\partial_{M}}_{tjli}(x_0)\Big) c(\widetilde{E}_{t}),
        ~~~~~~~~~~~~~{\rm if }~i,j<n; \\[2pt]   \\
       \Big(\frac{3}{4}(h'(0))^{2}-\frac{1}{2}h''(0)\Big)\sum_{j<n}\xi_{j}c(\widetilde{E}_{j}),~~~~~~~~~~~~~~{\rm if }~i=j=n,
       \end{array}
    \right.
\end{eqnarray}
where $\xi=\xi'+\xi_{n}\texttt{d}x_{n}$.
\end{lem}
\begin{proof}
From proposition 1.28 in \cite{BGV}, we have
 \begin{equation}
g_{ij}(x)\sim \delta_{ij}-\frac{1}{3}\sum_{kl}R^{\partial_{M}}_{ikjl}(x_{0})x^{k}x^{l}
 +\sum_{|\alpha|\geq 3}(\partial^{\alpha}g_{ij})(x_{0})\frac{x^{\alpha}}{\alpha !}.
\end{equation}
When $i,j<n$, we obtain
\begin{eqnarray}
\partial_{x_i}\partial_{x_j}(|\xi|_{g^M}^2)(x_0)&=&\partial_{x_i}\partial_{x_j}\Big(h(x_{n})|\xi'|^{2}+\xi_{n}^{2}\Big)(x_0)\nonumber\\
                    &=&\partial_{x_i}\partial_{x_j}(|\xi'|^{2})(x_0)\nonumber\\
                    &=&\partial_{x_i}\partial_{x_j}(g_{\alpha\beta}\xi_{\alpha}\xi_{\beta})(x_0)\nonumber\\
                    &=&\partial_{x_i}\partial_{x_j}\Big(\delta_{\alpha\beta}
                     -\frac{1}{3}\sum_{kl}R^{\partial_{M}}_{\alpha k \beta l}(x_{0})x^{k}x^{l}
                     +\cdots\Big)(x_0)\xi_{\alpha}\xi_{\beta}\nonumber\\
              &=&-\frac{1}{3} \sum_{\alpha ,\beta <n}\Big(R^{\partial_{M}}_{i\alpha j\beta}(x_0)
              +R^{\partial_{M}}_{i\beta j\alpha}(x_0)\Big)\xi_{\alpha}\xi_{\beta}.
\end{eqnarray}
When $i<n,j=n;~ or~ i=n,j<n$, from lemma 3.3, we get $\partial_{x_i}\partial_{x_j}(|\xi|_{g^M}^2)(x_0)=0$. When
$i=j=n$, then $\partial_{x_n}\partial_{x_n}(|\xi|_{g^M}^2)(x_0)\Big|_{|\xi'|=1}
=\partial_{x_n}\partial_{x_n}\Big(h(x_{n})|\xi'|^{2}+\xi_{n}^{2}\Big)(x_0)\Big|_{|\xi'|=1}=h''(0).$

On the other hand,  let $H_{ij}=\omega_{i}(\frac{\partial}{\partial x_{j}})$, from lemma 1.7.10 in \cite{Y}, we have
 \begin{equation}
H_{ij}(x)\sim \delta_{ij}-\frac{1}{6}\sum_{kl}R^{\partial_{M}}_{ikjl}(x_{0})x^{k}x^{l}+\cdots.
\end{equation}
When $i,j<n$, we obtain
\begin{eqnarray}
\partial_{x_i}\partial_{x_j}(c(\xi))(x_0)&=&\partial_{x_i}\partial_{x_j}\Big(c(\xi')+\xi_{n}c(\texttt{d}x_{n})\Big)(x_0)\nonumber\\
                    &=&\partial_{x_i}\partial_{x_j}\Big(\sum_{l}\xi_{l}c(\texttt{d}x_{l})\Big)(x_0)\nonumber\\
                    &=&\partial_{x_i}\partial_{x_j}\Big(\sum_{l,s,t}\xi_{l}(g^{sl}H_{st} )c(\widetilde{E}_{t}) \Big)(x_0)\nonumber\\
                                          &=&\sum_{l,t}\xi_{l}\Big(\partial_{x_i}\partial_{x_j}g^{sl}\delta_{s}^{t}
                      +\delta_{s}^{l}\partial_{x_i}\partial_{x_j}H_{st} \Big)(x_0) c(\widetilde{E}_{t})\nonumber\\
                      &=&\sum_{l,t}\xi_{l}\Big(-\partial_{x_i}\partial_{x_j}g_{sl}\delta_{s}^{t}
                      +\delta_{s}^{l}\partial_{x_i}\partial_{x_j}H_{st} \Big)(x_0) c(\widetilde{E}_{t})\nonumber\\
                     &=&-\sum_{l,t}\xi_{l}\partial_{x_i}\partial_{x_j}\Big(\delta_{lt}
                     -\frac{1}{3}\sum_{\alpha \beta}R^{\partial_{M}}_{t\alpha l\beta }(x_{0})x^{\alpha}x^{\beta}+\cdots\Big)(x_0) c(\widetilde{E}_{t})\nonumber\\
                     &&+\sum_{l,t}\xi_{l}\partial_{x_i}\partial_{x_j}\Big(\delta_{lt}
                     -\frac{1}{6}\sum_{kl}R^{\partial_{M}}_{l\alpha t\beta}(x_{0})x^{\alpha}x^{\beta}+\cdots \Big)(x_0) c(\widetilde{E}_{t})
                       \nonumber\\
                    &=&\frac{1}{6}\sum_{l,t<n}\xi_{l}\Big(R^{\partial_{M}}_{tilj}(x_0)
                    +R^{\partial_{M}}_{tjli}(x_0)\Big) c(\widetilde{E}_{t}),
\end{eqnarray}
where we have used the following fact of lemma A.1 in \cite{Wa3}
 \begin{equation}
c(dx_j)=\sum_{1\leq
i,s<n}\frac{1}{\sqrt{h(x_n)}}g^{s,j}H_{s,i}c(\widetilde {E_i})
+\sum_{i=s=n}g^{n,j}c(\widetilde{E_n}).
\end{equation}
If $i<n,j=n;~ or~ i=n,j<n$, from Lemma 3.3, we get $\partial_{x_i}\partial_{x_j}(c(\xi))(x_0)=0$.
When $i=j=n$, then
 \begin{equation}
\partial_{x_n}\partial_{x_n}(c(\xi))(x_0)=\sum_{j<n}\xi_{j}\Big(\frac{1}{\sqrt{h(x_{n})}}\Big)''c(\widetilde{E}_{j})
=\Big(\frac{3}{4}(h'(0))^{2}-\frac{1}{2}h''(0)\Big)\sum_{j}\xi_{j}c(\widetilde{E}_{j}).
\end{equation}
\end{proof}

Then an application of Lemma 2.3 in \cite{Wa3} shows
\begin{lem}The following identity holds:
\begin{eqnarray}
p_{0}(x_{0})&=&-h'(0)c(\texttt{d}x_{n}); \\
\partial_{x_i}[ A_{\beta s}^{\alpha}](x_{0})&=&\left\{
       \begin{array}{c}
        \sum_{ \beta, i, s, \alpha}R^{\partial_{M}}_{ \beta i s \alpha}(x_0),  ~~~~~~~~~~  ~~~~~~~~~~~~~{\rm if }~i<n; \\[2pt] \\
       0,~~~~~~~~~~~~~~~~~~~~~~~~~~~~~~~~~~~~~~~~~~~~~{\rm if }~i=n.
       \end{array}
    \right.
\end{eqnarray}
\end{lem}
\begin{proof}
From lemma 5.7 in \cite{Zh}, we have
 \begin{equation}
 A_{\beta s}^{\alpha}=R_{\beta l s\alpha}x_{l}+O(|x|^{2}).
\end{equation}
Then we obtain $\partial_{x_i}[ A_{\beta s}^{\alpha}](x_{0})= \sum_{ \beta i s \alpha}R^{\partial_{M}}_{ \beta i s \alpha}(x_0).$
\end{proof}

\begin{lem}When $i<n$,
\begin{eqnarray}
\partial_{x_i}\Big(\sigma_{-2}(D^{-1})\Big)(x_{0})\Big|_{|\xi'|=1}&=& \frac{1}{8}\sum_{ \beta i s \alpha}R^{\partial_{M}}_{ \beta i s \alpha}(x_0)
 \frac{c(\xi)c(\widetilde{E}_{\beta})c(\widetilde{E}_{s})c(\widetilde{E}_{\alpha})c(\xi) }{(1+\xi_{n}^{2})^{2}}\nonumber\\
&&+\frac{1}{6} \sum_{l,t<n}\xi_{l}\Big(R^{\partial_{M}}_{tilj}(x_0)+R^{\partial_{M}}_{tjli}(x_0)\Big)
 \frac{c(\xi)c(\texttt{d}x_{j})c(\widetilde{E}_{t}) }{(1+\xi_{n}^{2})^{2}} \nonumber\\
 &&+\frac{1}{3} \sum_{\alpha ,\beta <n}\Big(R^{\partial_{M}}_{i\alpha j\beta}(x_0)+R^{\partial_{M}}_{i\beta j\alpha}(x_0)\Big)
 \xi_{\alpha}\xi_{\beta}\frac{c(\xi)c(\texttt{d}x_{j})c(\xi) }{(1+\xi_{n}^{2})^{3}}.
\end{eqnarray}
When $i=n$,
\begin{eqnarray}
\partial_{x_n}\Big(\sigma_{-2}(D^{-1})\Big)(x_{0})\Big|_{|\xi'|=1}&=&\Big( \frac{-h' }{(1+\xi_{n}^{2})^{2}}+ \frac{-h' }{(1+\xi_{n}^{2})^{3}}\Big)
   \partial_{x_n}[c(\xi')](x_0)c(\texttt{d}x_{n})c(\xi)\nonumber\\
   &&+ \Big( \frac{(h')^{2}-h'' }{(1+\xi_{n}^{2})^{2}}+\frac{2(h')^{2}-h'' }{(1+\xi_{n}^{2})^{3}}+\frac{3(h')^{2} }{(1+\xi_{n}^{2})^{4}}\Big)
   c(\xi)c(\texttt{d}x_{n})c(\xi)\nonumber\\
   &&+ \Big( \frac{-h' }{(1+\xi_{n}^{2})^{2}}+ \frac{-3h' }{(1+\xi_{n}^{2})^{3}}\Big)
   c(\xi)c(\texttt{d}x_{n})\partial_{x_n}[c(\xi')](x_0)\nonumber\\
   &&+\frac{1 }{(1+\xi_{n}^{2})^{2}}\partial_{x_n}[c(\xi')](x_0)c(\texttt{d}x_{n}) \partial_{x_n}[c(\xi')](x_0)\nonumber\\
   &&+\Big(\frac{3}{4}(h'(0))^{2}-\frac{1}{2}h''(0)\Big)\frac{1 }{(1+\xi_{n}^{2})^{3}}c(\xi)c(\texttt{d}x_{n})c(\xi').
   \end{eqnarray}
\end{lem}
\begin{proof}
When $i<n$, from lemma 3.3 and $\partial_{x_i} (c(\texttt{d}x_{j}))(x_{0})$=0, we get
\begin{eqnarray}
\partial_{x_i}\Big(\sigma_{-2}(D^{-1})\Big)(x_{0})&=& \frac{1}{8}\sum_{ \beta i s \alpha}\partial_{x_i}[ A_{\beta s}^{\alpha}](x_{0})
 \frac{c(\xi)c(\widetilde{E}_{\beta})c(\widetilde{E}_{s})c(\widetilde{E}_{\alpha})c(\xi) }{(1+\xi_{n}^{2})^{2}}\nonumber\\
&&+\sum_{j<n}c(\xi)c(\texttt{d}x_{j})\partial_{x_i}\partial_{x_j}(c(\xi))(x_0)
 \frac{1 }{(1+\xi_{n}^{2})^{2}} \nonumber\\
 &&-\sum_{j<n}c(\xi)c(\texttt{d}x_{j})c(\xi)\partial_{x_i}\partial_{x_j}(|\xi|_{g^M}^2)(x_0)
\frac{1}{(1+\xi_{n}^{2})^{3}}.
\end{eqnarray}
By lemma 3.4, we obtain (3.28).  Similarly, the conclusion (3.29) then follows easily.
\end{proof}

Let us now consider the $q_{-3}$. From (3.8) and the proof of lemma 3.6, we know that
\begin{lem}The following identity holds:
\begin{eqnarray}
q_{-3}(x_{0})\Big|_{|\xi'|=1}&=& \frac{-i(h')^{2} }{(1+\xi_{n}^{2})^{3}}c(\xi)c(\texttt{d}x_{n})c(\xi)c(\texttt{d}x_{n})c(\xi) \nonumber\\
&&+\frac{ih' }{(1+\xi_{n}^{2})^{3}}c(\xi)c(\texttt{d}x_{n})c(\xi)c(\texttt{d}x_{n}) \partial_{x_n}[c(\xi')](x_0) \nonumber\\
&&+\frac{-i(h')^{2} }{(1+\xi_{n}^{2})^{4}}c(\xi)c(\texttt{d}x_{n})c(\xi)c(\texttt{d}x_{n})c(\xi) \nonumber\\
&&+\frac{1}{8}\sum_{ \beta i s \alpha}R^{\partial_{M}}_{ \beta i s \alpha}(x_0)
 \frac{ -i}{(1+\xi_{n}^{2})^{3}}c(\xi)c(\widetilde{E}_{i})c(\xi)c(\widetilde{E}_{\beta})c(\widetilde{E}_{s})c(\widetilde{E}_{\alpha})c(\xi)\nonumber\\
&&+\frac{1}{6} \sum_{l,t<n}\xi_{l}\Big(R^{\partial_{M}}_{tilj}(x_0)+R^{\partial_{M}}_{tjli}(x_0)\Big)
 \frac{-i }{(1+\xi_{n}^{2})^{3}}c(\xi)c(\widetilde{E}_{i})c(\xi)c(\texttt{d}x_{j})c(\widetilde{E}_{t}) \nonumber\\
 &&+\frac{1}{3} \sum_{\alpha ,\beta <n}\Big(R^{\partial_{M}}_{i\alpha j\beta}(x_0)+R^{\partial_{M}}_{i\beta j\alpha}(x_0)\Big)
 \xi_{\alpha}\xi_{\beta}\frac{-i }{(1+\xi_{n}^{2})^{4}}c(\xi)c(\widetilde{E}_{i}) c(\xi)c(\texttt{d}x_{j})c(\xi)\nonumber\\
 &&+i\Big( \frac{h' }{(1+\xi_{n}^{2})^{3}}+ \frac{h' }{(1+\xi_{n}^{2})^{4}}\Big)
    c(\xi)c(\texttt{d}x_{n})\partial_{x_n}[c(\xi')](x_0)c(\texttt{d}x_{n})c(\xi)\nonumber\\
   &&-i \Big( \frac{(h')^{2}-h'' }{(1+\xi_{n}^{2})^{3}}+\frac{2(h')^{2}-h'' }{(1+\xi_{n}^{2})^{4}}+\frac{3(h')^{2} }{(1+\xi_{n}^{2})^{5}}\Big)
    c(\xi)c(\texttt{d}x_{n})c(\xi)c(\texttt{d}x_{n})c(\xi)\nonumber\\
   &&+ i\Big( \frac{h' }{(1+\xi_{n}^{2})^{3}}+ \frac{3h' }{(1+\xi_{n}^{2})^{4}}\Big)
    c(\xi)c(\texttt{d}x_{n})c(\xi)c(\texttt{d}x_{n})\partial_{x_n}[c(\xi')](x_0)\nonumber\\
   &&+\frac{-i}{(1+\xi_{n}^{2})^{3}}c(\xi)c(\texttt{d}x_{n})\partial_{x_n}[c(\xi')](x_0)  c(\texttt{d}x_{n}) \partial_{x_n}[c(\xi')](x_0)\nonumber\\
   &&+\Big(\frac{3}{4}(h'(0))^{2}-\frac{1}{2}h''(0)\Big)\frac{-i }{(1+\xi_{n}^{2})^{4}} c(\xi)c(\texttt{d}x_{n})c(\xi)c(\texttt{d}x_{n})c(\xi').
\end{eqnarray}
\end{lem}

\begin{lem}The following identity holds:
\begin{equation}
\texttt{tr}\Big[\partial_{x_{n}}c(\xi')\times \partial_{x_{n}}c(\xi')\Big](x_{0})\Big|_{|\xi'|=1}=-\big(h'(0)\big)^{2}.
\end{equation}
\end{lem}
\begin{proof}
Let $\xi'=\sum_{j=1}^{4}\xi_{j}\texttt{d}x^{j}$, then
\begin{equation}
c(\xi')=\sum_{j=1}^{4}\xi_{j}c(\texttt{d}x^{j})=\sum_{j,l=1}^{4}\xi_{j}
 \langle \texttt{d}x^{j}, \widetilde{E}^{l}\rangle c(\widetilde{E}_{l}).
\end{equation}
Set
\begin{equation}
c(\widehat{\xi})=\partial_{x_{n}}(c(\xi'))=\sum_{j,l=1}^{4}\xi_{j}
 \partial_{x_{n}}\Big(\sqrt{h}\langle \texttt{d}x^{j}, E^{l}\rangle_{\partial_{M}} \Big)c(\widetilde{E}_{l}).
\end{equation}
Then
\begin{equation}
\texttt{tr}\Big[\partial_{x_{n}}c(\xi')\times \partial_{x_{n}}c(\xi')\Big]
= \texttt{tr}\Big[c(\widehat{\xi})c(\widehat{\xi})\Big](x_{0})=-|\widehat{\xi'}|^{2}\texttt{tr}(\texttt{id})
=-4|\widehat{\xi'}|^{2}.
\end{equation}
Where
\begin{eqnarray}
|\widehat{\xi'}|^{2}(x_{0})\Big|_{|\xi'|=1}&=&=\sum_{l=1}^{4} \Big[\sum_{j=1}^{4}\xi_{j}
 \partial_{x_{n}}\big(\sqrt{h}\big)\langle \texttt{d}x^{j}, E^{l}\rangle_{\partial_{M}}  \Big]^{2}(x_{0})\Big|_{|\xi'|=1}\nonumber\\
 &=&\Big(\partial_{x_{n}}\big(\sqrt{h}\big)\Big)^{2}|\xi'|^{2}(x_{0})\Big|_{|\xi'|=1}\nonumber\\
 &=&\frac{(h'(x_{n}))^{2}}{4h(x_{n})}(x_{0})
 =\frac{(h'(0))^{2}}{4}.
\end{eqnarray}
Substituting (3.36) into (3.35), we have proved this lemma.

\end{proof}
From the remark above, now we can compute $\Phi$ (see formula (2.20) for definition of $\Phi$).
Since the sum is taken over $-r-\ell+k+j+|\alpha|=4, \ r, \ell\leq-1$, then we have the $\int_{\partial_{M}}\Phi$ is the sum of the following
 fifteen cases:

\textbf{Case (1)}: \ $r=-1, \ \ell=-1, \ k=0, \ j=1, \ |\alpha|=1$

From (2.20), we have
\begin{equation}
\text{ Case \ (1)}=\frac{i}{2}\int_{|\xi'|=1}\int_{-\infty}^{+\infty}\sum_{|\alpha|=1}\text{trace}
\Big[\partial_{x_{n}}\partial_{\xi'}^{\alpha}\pi_{\xi_{n}}^{+}\sigma_{-1}(D^{-1})
\partial_{x'}^{\alpha}\partial_{\xi_{n}}^{2}\sigma_{-1}(D^{-1})\Big](x_{0})\texttt{d}\xi_{n}\sigma(\xi')\texttt{d}x' .
\end{equation}
By Lemma 3.3, for $i<n$, we have
\begin{equation}
\partial_{x_i}\sigma_{-1}(D^{-1})(x_0)=\partial_{x_i}\left(\frac{\sqrt{-1}c(\xi)}{|\xi|^2}\right)(x_0)=
\frac{\sqrt{-1}\partial_{x_i}[c(\xi)](x_0)}{|\xi|^2}
-\frac{\sqrt{-1}c(\xi)\partial_{x_i}(|\xi|^2)(x_0)}{|\xi|^4}=0.
\end{equation}
So Case (1) vanishes.

\textbf{Case (2)}: \  $r=-1, \ \ell=-1, \ k=0, \ j=2, \ |\alpha|=0$

From (2.20), we have
\begin{equation}
\text{ Case \ (2)}=\frac{i}{6}\int_{|\xi'|=1}\int_{-\infty}^{+\infty}\sum_{j=2}\text{trace}\Big[\partial_{x_{n}}^{2}\pi_{\xi_{n}}^{+}
\sigma_{-1}(D^{-1})\partial_{\xi_{n}}^{3}\sigma_{-1}(D^{-1})\Big](x_{0})\texttt{d}\xi_{n}\sigma(\xi')\texttt{d}x'.
\end{equation}
By Lemma 3.2, a simple computation shows
\begin{equation}
\partial_{\xi_{n}}^{3}\sigma_{-1}(D^{-1})(x_{0})\Big|_{|\xi'|=1}
=\frac{24\xi_{n}-24\xi_{n}^{3}}{(1+\xi_{n}^{2})^{4}}\sqrt{-1}c(\xi')+\frac{-6\xi_{n}^{4}
+36\xi_{n}^{2}-6}{(1+\xi_{n}^{2})^{4}}\sqrt{-1}c(\texttt{d}x_{n}),
\end{equation}
and
\begin{equation}
\partial_{x_{n}}\sigma_{-1}(D^{-1})=\frac{\sqrt{-1}\partial_{x_{n}}\big(c(\xi)\big)}{|\xi|^{2}}
-\frac{\sqrt{-1}c(\xi)\partial_{x_{n}}(|\xi|^{2})}{|\xi|^{4}}.
\end{equation}
From Lemma 3.2, Lemma 3.3 and Lemma 3.4, we obtain
\begin{eqnarray}
\partial_{x_{n}}^{2}\sigma_{-1}(D^{-1})(x_{0})\Big|_{|\xi'|=1}&=&\partial_{x_{n}}\Big(\frac{\sqrt{-1}\partial_{x_{n}}\big(c(\xi))}{|\xi|^{2}}
-\frac{\sqrt{-1}c(\xi)\partial_{x_{n}}(|\xi|^{2})}{|\xi|^{4}}\Big)  (x_{0})\Big|_{|\xi'|=1}\nonumber\\
&=&\bigg\{\frac{\sqrt{-1}\partial_{x_{n}}^{2}(c(\xi))}{|\xi|^{2}}
  -\frac{\sqrt{-1}\partial_{x_{n}}\big(c(\xi))  \partial_{x_{n}}(|\xi|^{2})}{|\xi|^{4}}
\nonumber\\
&& -\frac{\sqrt{-1}\Big(\partial_{x_{n}}(c(\xi))\partial_{x_{n}}(|\xi|^{2})
+c(\xi)\partial^{2}_{x_{n}}(|\xi|^{2})
\Big)}{|\xi|^{4}}\nonumber\\
&&+\frac{2\sqrt{-1}c(\xi)\partial_{x_{n}}(|\xi|^{2})\partial_{x_{n}}(|\xi|^{2})}{|\xi|^{6}}\bigg\}  (x_{0})\Big|_{|\xi'|=1}\nonumber\\
&=&\sqrt{-1}\frac{\Big(\frac{3}{4}(h'(0))^{2}-\frac{1}{2}h''(0)\Big)c(\xi')-h''(0)c(\xi)-2h'\partial_{x_{n}}(c(\xi'))}
{(1+\xi_{n}^{2})^{2}}\nonumber\\
&&+\sqrt{-1}\frac{2(h'(0))^{2}c(\xi)}{(1+\xi_{n}^{2})^{3}}.
\end{eqnarray}
By (2.15) and the Cauchy integral formula, then
\begin{eqnarray}
\pi^+_{\xi_n}\left[\frac{c(\xi)}{(1+\xi_n^2)^2}\right]&=&\pi^+_{\xi_n}\left[\frac{c(\xi')+\xi_nc(dx_n)}{(1+\xi_n^2)^2}\right] \nonumber\\
&=&\frac{1}{2\pi i}\lim_{u\rightarrow
0^-}\int_{\Gamma^+}\frac{\frac{c(\xi')+\eta_nc(dx_n)}{(\eta_n+i)^2(\xi_n+iu-\eta_n)}}
{(\eta_n-i)^2}d\eta_n   \nonumber\\
&=&\left[\frac{c(\xi')+\eta_nc(dx_n)}{(\eta_n+i)^2(\xi_n-\eta_n)}\right]^{(1)}\bigg|_{\eta_n=i}  \nonumber\\
&=&-\frac{ic(\xi')}{4(\xi_n-i)}-\frac{c(\xi')+ic(dx_n)}{4(\xi_n-i)^2}.
\end{eqnarray}
Similarly, we obtain
\begin{equation}
\pi^+_{\xi_n}\left[\frac{1}{(1+\xi_n^2)^{2}}\right]=\frac{-2-i\xi_n }{4(\xi_n-i)^2},~~~
\pi^+_{\xi_n}\left[\frac{c(\xi)}{(1+\xi_n^2)^{3}}\right]=\frac{-3i\xi_n^{2} -9\xi_n+8i}{16(\xi_n-i)^{3}}c(\xi')
+\frac{-i\xi_n-3}{16(\xi_n-i)^{3}}c(dx_n).
\end{equation}
From the remark above, it is easy to see
\begin{eqnarray}
\partial_{x_{n}}^{2}\pi_{\xi_{n}}^{+}\sigma_{-1}(D^{-1})(x_{0})\Big|_{|\xi'|=1}
&=&\Big(\frac{3}{4}(h'(0))^{2}-\frac{1}{2}h''(0)\Big)\frac{c(\xi')}{2(\xi_n-i)}
-h'(0) \frac{\xi_n -2i }{2(\xi_n-i)^{2}}\partial_{x_{n}}\big(c(\xi')\big) \nonumber\\
&&-h''(0)\Big[\frac{\xi_n -2i }{4(\xi_n-i)^{2}}c(\xi')+\frac{1 }{4(\xi_n-i)^2}c(dx_n)  \Big] \nonumber\\
&&+2i(h'(0))^{2}\Big[\frac{-3i\xi_n^{2} -9\xi_n+8i}{16(\xi_n-i)^{3}}c(\xi')
+\frac{-i\xi_n-3}{16(\xi_n-i)^{3}}c(dx_n)  \Big].
\end{eqnarray}
By the relation of the Clifford action and $\texttt{tr}AB=\texttt{tr}BA$, then
\begin{eqnarray}
&&\texttt{tr}[c(\xi')c(\texttt{d}x_{n})]=0; \ \texttt{tr}[c(\texttt{d}x_{n})^{2}]=-4;\ \texttt{tr}[c(\xi')^{2}](x_{0})|_{|\xi'|=1}=-4;\nonumber\\
&&\texttt{tr}[\partial_{x_{n}}[c(\xi')]c(\texttt{d}x_{n})]=0; \ \texttt{tr}[\partial_{x_{n}}c(\xi')\times c(\xi')](x_{0})|_{|\xi'|=1}=-2h'(0).
\end{eqnarray}
 For more trace expansions, we can see \cite{GS}.
From (3.40), (3.45) and (3.46) and direct computations, we obtain
\begin{eqnarray}
&&\text{trace}\Big[\partial_{x_{n}}^{2}\pi_{\xi_{n}}^{+}
\sigma_{-1}(D^{-1})\partial_{\xi_{n}}^{3}\sigma_{-1}(D^{-1})\Big](x_{0})\Big|_{|\xi'|=1}\nonumber\\
&=&\big(h'(0)\big)^{2}\frac{3(33\xi_{n}^{5}-75i\xi_{n}^{4}-94\xi_{n}^{3}+90i\xi_{n}^{2}+57\xi_{n}-3i)}
 {2(\xi_{n}-i)^{3}(1+\xi_{n}^{2})^{4}}\nonumber\\
 &&+h''(0)\frac{6(-9\xi_{n}^{4}+12i\xi_{n}^{3}+14\xi_{n}^{2}-12i\xi_{n}-1)}
 {2(\xi_{n}-i)^{2}(1+\xi_{n}^{2})^{4}}.
\end{eqnarray}
Therefore
\begin{eqnarray}
\text{ Case \ (2) }
&=&-\frac{1}{6}\big(h'(0)\big)^{2}\int_{|\xi'|=1}\int_{-\infty}^{+\infty}\frac{3(33\xi_{n}^{5}-75i\xi_{n}^{4}-94\xi_{n}^{3}
+90i\xi_{n}^{2}+57\xi_{n}-3i)} {2(\xi_{n}-i)^{3}(1+\xi_{n}^{2})^{4}}\texttt{d}\xi_{n}\sigma(\xi')\texttt{d}x'\nonumber\\
&&-\frac{1}{6}h''(0)\int_{|\xi'|=1}\int_{-\infty}^{+\infty}\frac{6(-9\xi_{n}^{4}+12i\xi_{n}^{3}+14\xi_{n}^{2}-12i\xi_{n}-1)}
 {2(\xi_{n}-i)^{2}(1+\xi_{n}^{2})^{4}}\texttt{d}\xi_{n}\sigma(\xi')\texttt{d}x'\nonumber\\
 &=&-\frac{1}{6}\big(h'(0)\big)^{2} \Omega_{3}\int_{\Gamma^{+}}\frac{3(33\xi_{n}^{5}-75i\xi_{n}^{4}-94\xi_{n}^{3}
+90i\xi_{n}^{2}+57\xi_{n}-3i)} {2(\xi_{n}-i)^{3}(1+\xi_{n}^{2})^{4}}\texttt{d}\xi_{n}\texttt{d}x'\nonumber\\
&&-\frac{1}{6}h''(0)\Omega_{3}\int_{\Gamma^{+}}\frac{6(-9\xi_{n}^{4}+12i\xi_{n}^{3}+14\xi_{n}^{2}-12i\xi_{n}-1)}
 {2(\xi_{n}-i)^{2}(1+\xi_{n}^{2})^{4}}\texttt{d}\xi_{n}\texttt{d}x'\nonumber\\
&=&-\frac{1}{6}\big(h'(0)\big)^{2}  \frac{ 2\pi i }{6!}\bigg[\frac{3(33\xi_{n}^{5}-75i\xi_{n}^{4}-94\xi_{n}^{3}
+90i\xi_{n}^{2}+57\xi_{n}-3i)}{2(\xi_{n}+i)^{4}}\bigg]^{(6)}\bigg|_{\xi_{n}=i}\Omega_{3}\texttt{d}x'  \nonumber\\
&&-\frac{1}{6}h''(0) \frac{ 2\pi i }{5!}\bigg[\frac{6(-9\xi_{n}^{4}+12i\xi_{n}^{3}+14\xi_{n}^{2}-12i\xi_{n}-1)}
{(\xi_{n}+i)^{4}}\bigg]^{(5)}\bigg|_{\xi_{n}=i}\Omega_{3}\texttt{d}x'  \nonumber\\
&=&\Big(\frac{29}{64}\big(h'(0)\big)^{2}-\frac{3}{8}h''(0)\Big)\pi\Omega_{3}\texttt{d}x',
\end{eqnarray}
where $\Omega_{3}$ is the canonical volume of $S^{3}$.

\textbf{Case (3)}: \ $r=-1, \ \ell=-1, \ k=0, \ j=0, \ |\alpha|=2$

From (2.20), we have
\begin{equation}
\text{ Case \ (3)}=\frac{i}{2}\int_{|\xi'|=1}\int_{-\infty}^{+\infty}\sum_{|\alpha|=2}
\text{trace}\Big[\partial_{\xi'}^{\alpha}\pi_{\xi_{n}}^{+}\sigma_{-1}(D^{-1})
\partial_{x'}^{\alpha}\partial_{\xi_{n}}\sigma_{-1}(D^{-1})\Big](x_{0})\texttt{d}\xi_{n}\sigma(\xi')\texttt{d}x' .
\end{equation}
By Lemma 3.2, a simple computation shows
\begin{eqnarray}
\partial_{\xi'}^{\alpha}\sigma_{-1}(D^{-1})(x_{0})\Big|_{|\xi'|=1}
&=&\partial_{\xi_{j}}\partial_{\xi_{i}}\sigma_{-1}(D^{-1})(x_{0})\Big|_{|\xi'|=1}  \nonumber\\
&=&\partial_{\xi_{j}}\Big(\frac{\sqrt{-1}c(\texttt{d}x_{i})}{(1+\xi_{n}^{2})}
-\frac{2\sqrt{-1}\xi_{i}c(\xi) }{(1+\xi_{n}^{2})^{2}}\Big) \nonumber\\
&=&\sqrt{-1}\Big(\frac{-2\xi_{j}c(\texttt{d}x_{i})-2\delta_{j}^{i}c(\xi)-2\xi_{i}c(\texttt{d}x_{j})}{(1+\xi_{n}^{2})^{2}}
 +\frac{8\xi_{i}\xi_{j}c(\xi) }{(1+\xi_{n}^{2})^{3}}\Big).
\end{eqnarray}
By (3.43) and (3.44), we obtain
\begin{eqnarray}
\pi_{\xi_{n}}^{+}\partial_{\xi'}^{\alpha}\sigma_{-1}(D^{-1})(x_{0})\Big|_{|\xi'|=1}
&=&\xi_{j}\frac{2i-\xi_{n}}{2(\xi_{n}-i)^{2}}c(\texttt{d}x_{i})+\frac{2i-\xi_{n}}{2(\xi_{n}-i)^{2}}c(\xi')\delta_{j}^{i}
+\frac{-1}{2(\xi_{n}-i)^{2}}c(\texttt{d}x_{n}) \delta_{j}^{i} \nonumber\\
&&+\xi_{i}\frac{2i-\xi_{n}}{2(\xi_{n}-i)^{2}}c(\texttt{d}x_{j})+\xi_{i}\xi_{j}\frac{3\xi_{n}^{2}-9i\xi_{n}-8}{2(\xi_{n}-i)^{3}}c(\xi')
\nonumber\\
&&+\xi_{i}\xi_{j}\frac{\xi_{n}-3i}{2(\xi_{n}-i)^{3}}c(\texttt{d}x_{n}).
\end{eqnarray}
On the other hand, by Lemma 3.2, Lemma 3.3 and Lemma 3.4, we obtain
\begin{eqnarray}
\partial_{x'}^{\alpha}\sigma_{-1}(D^{-1})(x_{0})\Big|_{|\xi'|=1}
&=&\partial_{x_{i}}\partial_{x_{j}}\sigma_{-1}(D^{-1})(x_{0})\Big|_{|\xi'|=1}\nonumber\\
&=&\partial_{x_{i}}\Big(\frac{\sqrt{-1}\partial_{x_{j}}\big(c(\xi))}{|\xi|^{2}}
-\frac{\sqrt{-1}c(\xi)\partial_{x_{j}}(|\xi|^{2})}{|\xi|^{4}}\Big)  (x_{0})\Big|_{|\xi'|=1}\nonumber\\
&=&\Big(\frac{\sqrt{-1}\partial_{x_{i}}\partial_{x_{j}}(c(\xi))}{|\xi|^{2}}
  -\frac{\sqrt{-1}c(\xi) \partial_{x_{i}}\partial_{x_{j}}(|\xi|^{2})}{|\xi|^{4}}\Big)
 (x_{0})\Big|_{|\xi'|=1}\nonumber\\
&=&\frac{1}{6}\sum_{l,t<n}\xi_{l}\Big(R^{\partial_{M}}_{tilj}(x_0)+R^{\partial_{M}}_{tjli}(x_0)\Big)
 \frac{i }{(1+\xi_{n}^{2})}c(\widetilde{E}_{t})
\nonumber\\
  &&+\frac{1}{3} \sum_{\alpha ,\beta <n}\Big(R^{\partial_{M}}_{i\alpha j\beta}(x_0)+R^{\partial_{M}}_{i\beta j\alpha}(x_0)\Big)
 \xi_{\alpha}\xi_{\beta}\frac{i }{(1+\xi_{n}^{2})^{2}}c(\xi).
\end{eqnarray}
Hence in this case,
\begin{eqnarray}
\partial_{x'}^{\alpha}\partial_{\xi_{n}}\sigma_{-1}(D^{-1})(x_{0})\Big|_{|\xi'|=1}
&=&\sum_{l,t<n}\xi_{l}\Big(R^{\partial_{M}}_{tilj}(x_0)+R^{\partial_{M}}_{tjli}(x_0)\Big)
 \frac{-2i\xi_{n}  }{(1+\xi_{n}^{2})^{2}}c(\widetilde{E}_{t})
\nonumber\\
&&+\sum_{\alpha ,\beta <n}\Big(R^{\partial_{M}}_{i\alpha j\beta}(x_0)+R^{\partial_{M}}_{i\beta j\alpha}(x_0)\Big)
 \xi_{\alpha}\xi_{\beta}\frac{-4i\xi_{n} }{3(1+\xi_{n}^{2})^{3}}c(\xi')\nonumber\\
  &&+\sum_{\alpha ,\beta <n}\Big(R^{\partial_{M}}_{i\alpha j\beta}(x_0)+R^{\partial_{M}}_{i\beta j\alpha}(x_0)\Big)
 \xi_{\alpha}\xi_{\beta}\frac{i-3i\xi_{n}^{2} }{3(1+\xi_{n}^{2})^{3}}c(\texttt{d}x_{n}).
\end{eqnarray}
Considering for $i<n$,  $\int_{|\xi'|=1}\{\xi_{i_1}\xi_{i_2}\cdots\xi_{i_{2d+1}}\}\sigma(\xi')=0$.
 From (3.46), (3.51), (3.53) and direct computations, we obtain
\begin{eqnarray}
&&\text{trace}\Big[\partial_{\xi'}^{\alpha}\pi_{\xi_{n}}^{+}\sigma_{-1}(D^{-1})
\partial_{x'}^{\alpha}\partial_{\xi_{n}}\sigma_{-1}(D^{-1})\Big](x_{0})\nonumber\\
&=&\sum_{l,j,t <n}R^{\partial_{M}}_{tjlt}(x_0) \xi_{l}\xi_{j}
\frac{-4i\xi_{n}^{2}-8\xi_{n}} {3(\xi_{n}-i)^{2}(1+\xi_{n}^{2})^{2}}\nonumber\\
 &&+\sum_{t,i,l <n}R^{\partial_{M}}_{tili}(x_0) \xi_{l}\xi_{t}\frac{-4i\xi_{n}^{4}-8\xi_{n}^{3}-32i\xi_{n}^{2}-40\xi_{n}+4i}
 {3(\xi_{n}-i)^{2}(1+\xi_{n}^{2})^{3}}.
\end{eqnarray}
Then an application of (16) in \cite{Ka} shows
\begin{equation}
\int \xi^{\mu}\xi^{\nu}=\frac{1}{4}[^{\mu\nu}],~~\int \xi^{\mu}\xi^{\nu}\xi^{\alpha}\xi^{\beta}=\frac{1}{3\cdot 2^{3}}[^{\mu\nu\alpha\beta}],
\end{equation}
where $[^{\mu\nu\alpha\beta}]$ stands for the sum of products of $g^{\alpha\beta}$ determined by all "pairings"
of $\mu\nu\alpha\beta$.
Using the integration over $S^{3}$ and the shorthand $ \int=\frac{1}{\pi^{2}}\int_{S^{3}}d^{3}\nu$, we obtain  $\Omega_{3}=2\pi^{2} $.
Let $s_{\partial_{M}}$ is the scalar curvature $\partial_{M}$, then
\begin{eqnarray}
&&\sum_{\alpha, i, s <n}R^{\partial_{M}}_{\alpha i s\alpha }(x_0)\int_{|\xi'|=1} \xi_{i}\xi_{s} \sigma(\xi')
=\sum_{\alpha, i, s <n}R^{\partial_{M}}_{\alpha i s\alpha}(x_0)\frac{\pi^{2}}{2}\delta_{i}^{s}=-\frac{1}{4}s_{\partial_{M}}\Omega_{3},\\
 &&\sum_{i,\alpha, j,\beta <n}R^{\partial_{M}}_{i\alpha j \beta }(x_0)\int_{|\xi'|=1} \xi_{\alpha}\xi_{\beta}\xi_{i}\xi_{j} \sigma(\xi')
 = \frac{\pi^{2}}{2}\sum_{i,\alpha, j,\beta <n}R^{\partial_{M}}_{i\alpha j \beta }(x_0)
\Big(\delta_{\alpha}^{\beta}\delta_{i}^{j}+\delta_{\alpha}^{i}\delta_{\beta}^{j}
 +\delta_{\alpha}^{j}\delta_{\beta}^{i}\Big)=0.
\end{eqnarray}
Therefore
\begin{eqnarray}
\text{ Case \ (3) }
&=&\frac{i}{2}\int_{|\xi'|=1}\sum_{l,j,t <n}R^{\partial_{M}}_{tjlt}(x_0) \xi_{l}\xi_{j} \sigma(\xi')
\int_{-\infty}^{+\infty}\frac{-4i\xi_{n}^{2}-8\xi_{n}} {3(\xi_{n}-i)^{2}(1+\xi_{n}^{2})^{2}}
\texttt{d}\xi_{n}\texttt{d}x'\nonumber\\
&&+\frac{i}{2}\int_{|\xi'|=1}\sum_{t,i,l <n}R^{\partial_{M}}_{tili}(x_0) \xi_{l}\xi_{t} \sigma(\xi')
\int_{-\infty}^{+\infty}\frac{-4i\xi_{n}^{4}-8\xi_{n}^{3}-32i\xi_{n}^{2}-40\xi_{n}+4i}
 {3(\xi_{n}-i)^{2}(1+\xi_{n}^{2})^{3}}\texttt{d}\xi_{n}\texttt{d}x'\nonumber\\
&=&\frac{i}{2}\sum_{l,j,t <n}R^{\partial_{M}}_{tjlt}(x_0) \frac{\pi^{2}}{2}\delta_{l}^{j}
\frac{ 2\pi i }{3!}\bigg[\frac{-4i\xi_{n}^{2}-8\xi_{n}}
{3(\xi_{n}+i)^{2}}\bigg]^{(3)}\bigg|_{\xi_{n}=i}\Omega_{3}\texttt{d}x'
\nonumber\\
&&+\frac{i}{2}\sum_{t,i,l <n}R^{\partial_{M}}_{tili}(x_0) \frac{\pi^{2}}{2}\delta_{l}^{t}
\frac{ 2\pi i }{4!}\bigg[\frac{-4i\xi_{n}^{4}-8\xi_{n}^{3}-32i\xi_{n}^{2}-40\xi_{n}+4i}
{3(\xi_{n}+i)^{3}}\bigg]^{(4)}\bigg|_{\xi_{n}=i}\Omega_{3}\texttt{d}x' \nonumber\\
&=&\frac{i}{2}\Big(\frac{-i}{6}\pi^{3}\sum_{t,l <n}R^{\partial_{M}}_{tllt}(x_0)
 +\frac{-2i}{3}\pi^{3}\sum_{t,l <n}R^{\partial_{M}}_{tltl}(x_0)  \Big)\texttt{d}x' \nonumber\\
&=&\frac{1}{4}s_{\partial_{M}}\pi^{3}\texttt{d}x',
\end{eqnarray}
where $\sum_{t,l <n}R^{\partial_{M}}_{tltl}(x_0) $ is the scalar curvature $s_{\partial_{M}}$.

\textbf{Case (4)}: \ $r=-1, \ \ell=-1, \ k=1, \ j=1, \ |\alpha|=0$

From (2.20) and the Leibniz rule, we obtain
\begin{eqnarray}
\text{ Case \ (4)}&=&\frac{i}{6}\int_{|\xi'|=1}\int_{-\infty}^{+\infty}
\text{trace}\Big[\partial_{x_{n}}\partial_{\xi_{n}}\pi_{\xi_{n}}^{+}\sigma_{-1}(D^{-1})
\partial_{\xi_{n}}^{2}\partial_{x_{n}}\sigma_{-1}(D^{-1})\Big](x_{0})\texttt{d}\xi_{n}\sigma(\xi')\texttt{d}x' \nonumber\\
  &=&-\frac{i}{6}\int_{|\xi'|=1}\int_{-\infty}^{+\infty}
\text{trace}\Big[\partial_{x_{n}}\pi_{\xi_{n}}^{+}\sigma_{-1}(D^{-1})
\partial_{\xi_{n}}^{3}\partial_{x_{n}}\sigma_{-1}(D^{-1})\Big](x_{0})\texttt{d}\xi_{n}\sigma(\xi')\texttt{d}x'.
\end{eqnarray}
By (2.2.22) in \cite{Wa3}, we have
\begin{equation}
\pi^+_{\xi_n}\partial_{x_n}\sigma_{-1}(D^{-1})(x_0)|_{|\xi'|=1}=\frac{\partial_{x_n}[c(\xi')](x_0)}{2(\xi_n-i)}+\sqrt{-1}h'(0)
\left[\frac{ic(\xi')}{4(\xi_n-i)}+\frac{c(\xi')+ic(\texttt{d}x_n)}{4(\xi_n-i)^2}\right].
\end{equation}
From (3.42) and direct computations, we obtain
\begin{eqnarray}
\partial_{\xi_{n}}^{3}\partial_{x_{n}}\sigma_{-1}(D^{-1})(x_0)|_{|\xi'|=1}
&=&\frac{24i\xi_n-24i\xi_n^{3}}{(1+\xi_n^{2})^{4}}\partial_{x_n}[c(\xi')](x_0)
+\sqrt{-1}h'(0)\Big[\frac{8(15\xi_n^{3}-9\xi_n)}{(1+\xi_n^{2})^{5}}c(\xi')\nonumber\\
&&+\frac{12(5\xi_n^{4}-10\xi_n^{2}+1)}{(1+\xi_n^{2})^{5}}c(\texttt{d}x_n)\Big].
\end{eqnarray}
From Lemma 3.8, combining (3.46), (3.59) and (3.60), we obtain
\begin{eqnarray}
&&\text{trace}\Big[\partial_{x_{n}}\pi_{\xi_{n}}^{+}\sigma_{-1}(D^{-1})
\partial_{\xi_{n}}^{3}\partial_{x_{n}}\sigma_{-1}(D^{-1})\Big](x_{0})\nonumber\\
&=&(h'(0))^{2}
\frac{12(-\xi_{n}^{5}+5i\xi_{n}^{4}+10\xi_{n}^{3}-10i\xi_{n}^{2}-5\xi_{n}+i)} {(\xi_{n}-i)^{2}(1+\xi_{n}^{2})^{5}}.
\end{eqnarray}
Therefore
\begin{eqnarray}
\text{ Case \ (4) }
&=&\frac{i}{6}\big(h'(0)\big)^{2}\int_{|\xi'|=1}\int_{-\infty}^{+\infty}
\frac{12(-\xi_{n}^{5}+5i\xi_{n}^{4}+10\xi_{n}^{3}-10i\xi_{n}^{2}-5\xi_{n}+i)} {(\xi_{n}-i)^{2}(1+\xi_{n}^{2})^{5}}
\texttt{d}\xi_{n}\sigma(\xi')\texttt{d}x'\nonumber\\
&=&\frac{i}{6}\big(h'(0)\big)^{2}  \frac{ 2\pi i }{6!}\bigg[
\frac{12(-\xi_{n}^{5}+5i\xi_{n}^{4}+10\xi_{n}^{3}-10i\xi_{n}^{2}-5\xi_{n}+i)}
{(\xi_{n}+i)^{5}}\bigg]^{(6)}\bigg|_{\xi_{n}=i}\Omega_{3}\texttt{d}x'  \nonumber\\
&=&-\frac{5}{16}\big(h'(0)\big)^{2}\pi\Omega_{3}\texttt{d}x'.
\end{eqnarray}

\textbf{Case (5)}: \ $r=-1, \ \ell=-1, \ k=1, \ j=0, \ |\alpha|=1$

From (2.20), we have
\begin{equation}
\text{ Case \ (5)}=\frac{i}{2}\int_{|\xi'|=1}\int_{-\infty}^{+\infty}\sum_{|\alpha|=1}
\text{trace}\Big[\partial_{\xi'}^{\alpha}\partial_{\xi_{n}}\pi_{\xi_{n}}^{+}\sigma_{-1}(D^{-1})
\partial_{x'}^{\alpha}\partial_{\xi_{n}}\partial_{x_{n}}\sigma_{-1}(D^{-1})\Big](x_{0})\texttt{d}\xi_{n}\sigma(\xi')\texttt{d}x' .
\end{equation}
From Lemma 3.3 and Lemma 3.4, for $i<n$, we have
\begin{eqnarray}
\partial_{x_i}\partial_{x_{n}}\sigma_{-1}(D^{-1})(x_0)
&=&\partial_{x_i}\Big[\partial_{x_n}\Big(\frac{\sqrt{-1}c(\xi)}{|\xi|^2}\Big)\Big](x_0)\nonumber\\
&=&\partial_{x_i}\Big[\frac{\sqrt{-1}\partial_{x_n}(c(\xi))}{|\xi|^2}
-\frac{\sqrt{-1}c(\xi)\partial_{x_n}(|\xi|^2)}{|\xi|^4}\Big](x_0)\nonumber\\
&=&\sqrt{-1}\Big[\frac{\partial_{x_i}\partial_{x_n}(c(\xi))}{|\xi|^2}
-\frac{\partial_{x_n}[c(\xi)]\partial_{x_i}(|\xi|^2)}{|\xi|^4}\Big](x_0)\nonumber\\
&&-\sqrt{-1}\Big[\frac{\partial_{x_i}(c(\xi))\partial_{x_n}(|\xi|^2)}{|\xi|^4}
+\frac{c(\xi)\partial_{x_i}\partial_{x_n}(|\xi|^2)}{|\xi|^4}
-\frac{2c(\xi)\partial_{x_n}(|\xi|^2)\partial_{x_i}(|\xi|^2)}{|\xi|^6}\Big](x_0)\nonumber\\
&=&0.
\end{eqnarray}
Therefore Case (5) vanishes.

\textbf{Case (6)}: \ $r=-1, \ \ell=-1, \ k=2, \ j=0, \ |\alpha|=0$

From (2.20), we have
\begin{equation}
\text{ Case \ (6)}=\frac{i}{6}\int_{|\xi'|=1}\int_{-\infty}^{+\infty}\sum_{k=2}
\text{trace}\Big[\partial_{\xi_{n}}^{2}\pi_{\xi_{n}}^{+}\sigma_{-1}(D^{-1})
\partial_{\xi_{n}}\partial_{x_{n}}^{2}\sigma_{-1}(D^{-1})\Big](x_{0})\texttt{d}\xi_{n}\sigma(\xi')\texttt{d}x'.
\end{equation}
By the Leibniz rule, trace property and "++" and "-~-" vanishing
after the integration over $\xi_n$ in \cite{FGLS}, then
\begin{eqnarray}
&&\int^{+\infty}_{-\infty}{\rm trace}
\Big[\partial_{\xi_{n}}^{2}\pi_{\xi_{n}}^{+}\sigma_{-1}(D^{-1})
\partial_{\xi_{n}}\partial_{x_{n}}^{2}\sigma_{-1}(D^{-1})\Big]\texttt{d}\xi_n \nonumber\\
&=&-\int^{+\infty}_{-\infty}{\rm trace}
\Big[\partial_{\xi_{n}}^{3}\pi_{\xi_{n}}^{+}\sigma_{-1}(D^{-1})\partial_{x_{n}}^{2}\sigma_{-1}(D^{-1})\Big]\texttt{d}\xi_n \nonumber\\
&=&-\bigg[\int^{+\infty}_{-\infty}{\rm trace}
\Big[\partial_{\xi_{n}}^{3}\sigma_{-1}(D^{-1})\partial_{x_{n}}^{2}\sigma_{-1}(D^{-1})\Big]\texttt{d}\xi_n\nonumber\\
&&-\int^{+\infty}_{-\infty}{\rm trace}
\Big[\partial_{\xi_{n}}^{3}\pi_{\xi_{n}}^{-}\sigma_{-1}(D^{-1})\partial_{x_{n}}^{2}\pi_{\xi_{n}}^{+}\sigma_{-1}(D^{-1})\Big]
\texttt{d}\xi_n \bigg]\nonumber\\
&=&-\bigg[\int^{+\infty}_{-\infty}{\rm trace}
\Big[\partial_{\xi_{n}}^{3}\sigma_{-1}(D^{-1})\partial_{x_{n}}^{2}\sigma_{-1}(D^{-1})\Big]\texttt{d}\xi_n
-\int^{+\infty}_{-\infty}{\rm trace}
\Big[\partial_{\xi_{n}}^{3}\sigma_{-1}(D^{-1})\partial_{x_{n}}^{2}\pi_{\xi_{n}}^{+}\sigma_{-1}(D^{-1})\Big]\texttt{d}\xi_n \bigg]\nonumber\\
&=& \int^{+\infty}_{-\infty}{\rm trace}
\Big[\partial_{\xi_{n}}^{3}\sigma_{-1}(D^{-1})\partial_{x_{n}}^{2}\pi_{\xi_{n}}^{+}\sigma_{-1}(D^{-1})\Big]\texttt{d}\xi_n
-\int^{+\infty}_{-\infty}{\rm trace}
\Big[\partial_{\xi_{n}}^{3}\sigma_{-1}(D^{-1})\partial_{x_{n}}^{2}\sigma_{-1}(D^{-1})\Big]\texttt{d}\xi_n.
\end{eqnarray}
Combining these assertions, we see
\begin{equation}
\text{ Case \ (6)}=\text{ Case \ (2)}-\int^{+\infty}_{-\infty}{\rm trace}
\Big[\partial_{\xi_{n}}^{3}\sigma_{-1}(D^{-1})\partial_{x_{n}}^{2}\sigma_{-1}(D^{-1})\Big]\texttt{d}\xi_n.
\end{equation}
From Lemma 3.3, we have
\begin{equation}
\partial_{\xi_{n}}^{3}\sigma_{-1}(D^{-1})(x_{0})\Big|_{|\xi'|=1}
=\frac{24\xi_{n}-24\xi_{n}^{3}}{(1+\xi_{n}^{2})^{4}}\sqrt{-1}c(\xi')+\frac{-6\xi_{n}^{4}
+36\xi_{n}^{2}-6}{(1+\xi_{n}^{2})^{4}}\sqrt{-1}c(\texttt{d}x_{n}),
\end{equation}
and
\begin{eqnarray}
\partial_{x_{n}}^{2}\sigma_{-1}(D^{-1})(x_{0})\Big|_{|\xi'|=1}
&=&\frac{-2\sqrt{-1}h'(0)}{(1+\xi_{n}^{2})^{2}}\partial_{x_{n}}\big(c(\xi')\big)(x_{0})
  +\frac{2\xi_{n}(h'(0))^{2}-\xi_{n}(1+\xi_{n}^{2})h''(x_{0})}{(1+\xi_{n}^{2})^{3}}\sqrt{-1}c(\texttt{d}x_{n})\nonumber\\
  &&+\frac{(11+3\xi_{n}^{2})(h'(0))^{2}-(6+6\xi_{n}^{2})h''(x_{0})}{4(1+\xi_{n}^{2})^{3}}\sqrt{-1}c(\xi').
\end{eqnarray}
Combining (3.46), (3.69) and (3.70), we obtain
\begin{equation}
\text{trace}\Big[\partial_{\xi_{n}}^{3}\sigma_{-1}(D^{-1})\partial_{x_{n}}^{2}\sigma_{-1}(D^{-1})\Big](x_{0})
=(h'(0))^{2}\frac{24\xi_{n}(5-3\xi_{n}^{2})} {(1+\xi_{n}^{2})^{5}}+h''(0)\frac{24\xi_{n}(3\xi_{n}^{2}-5)} {(1+\xi_{n}^{2})^{5}}.
\end{equation}
We note that
\begin{equation}
\int_{-\infty}^{+\infty}\frac{24\xi_{n}(5-3\xi_{n}^{2})} {(1+\xi_{n}^{2})^{5}}\texttt{d}\xi_{n}
=\frac{ 2\pi i }{4!}\bigg[\frac{24\xi_{n}(5-3\xi_{n}^{2})}{(\xi_{n}+i)^{5}}\bigg]^{(4)}\Big|_{\xi_{n}=i}=0.
\end{equation}
Therefore
\begin{equation}
\text{ Case \ (6)}=\Big(\frac{29}{64}\big(h'(0)\big)^{2}-\frac{3}{8}h''(0)\Big)\pi\Omega_{3}\texttt{d}x'.
\end{equation}

\textbf{Case (7)}: \ $r=-1, \ \ell=-2, \ k=0, \ j=1, \ |\alpha|=0$

From (2.20) and the Leibniz rule, we obtain
\begin{eqnarray}
\text{ Case \ (7)}&=&\frac{1}{2}\int_{|\xi'|=1}\int_{-\infty}^{+\infty}
\text{trace}\Big[\partial_{\xi_{n}}\partial_{x_{n}}\pi_{\xi_{n}}^{+}\sigma_{-1}(D^{-1})
\partial_{\xi_{n}}\sigma_{-2}(D^{-1})\Big](x_{0})\texttt{d}\xi_{n}\sigma(\xi')\texttt{d}x' \nonumber\\
  &=&-\frac{1}{2}\int_{|\xi'|=1}\int_{-\infty}^{+\infty}
\text{trace}\Big[\partial_{\xi_{n}}^{2}\partial_{x_{n}}\pi_{\xi_{n}}^{+}\sigma_{-1}(D^{-1})
\sigma_{-2}(D^{-1})\Big](x_{0})\texttt{d}\xi_{n}\sigma(\xi')\texttt{d}x'.
\end{eqnarray}
By Lemma 3.3 and (2.2.22) in \cite{Wa3}, we have
\begin{equation}
\pi^+_{\xi_n}\partial_{x_n}\sigma_{-1}(D^{-1})(x_0)|_{|\xi'|=1}=\frac{\partial_{x_n}[c(\xi')](x_0)}{2(\xi_n-i)}+\sqrt{-1}h'(0)
\left[\frac{ic(\xi')}{4(\xi_n-i)}+\frac{c(\xi')+ic(\texttt{d}x_n)}{4(\xi_n-i)^2}\right].
\end{equation}
By direct computations, we obtain
\begin{equation}
\partial_{\xi_{n}}^{2}\partial_{x_{n}}\pi_{\xi_{n}}^{+}\sigma_{-1}(D^{-1})(x_0)|_{|\xi'|=1}
=\frac{1}{(\xi-1)^{3}}\partial_{x_n}[c(\xi')](x_0)
+h'(0)\Big[\frac{4i-\xi_n}{2(\xi-1)^{4}}c(\xi')
-\frac{3}{2(\xi-1)^{4}}c(\texttt{d}x_n)\Big].
\end{equation}
From Lemma 3.2 and Lemma 3.5, we have
\begin{equation}
\sigma_{-2}(D^{-1})(x_0)=\frac{-h'(0)c(\xi)c(\texttt{d}x_n)c(\xi)}{|\xi|^4}+\frac{c(\xi)}{|\xi|^6}c(\texttt{d}x_n)
\Big[\partial_{x_n}[c(\xi')](x_0)|\xi|^2-c(\xi)h'(0)|\xi|^2_{\partial M}\Big].
\end{equation}
Considering for $i<n$,  $\int_{|\xi'|=1}\{\xi_{i_1}\xi_{i_2}\cdots\xi_{i_{2d+1}}\}\sigma(\xi')=0$.
 From Lemma 3.8, combining (3.46), (3.76), (3.77) and direct computations, we obtain
\begin{equation}
\text{trace}\Big[\partial_{\xi_{n}}^{2}\partial_{x_{n}}\pi_{\xi_{n}}^{+}\sigma_{-1}(D^{-1})\sigma_{-2}(D^{-1})\Big](x_{0})
=(h'(0))^{2}\frac{3(-2\xi_{n}^{4}+3i\xi_{n}^{3}-3\xi_{n}^{2}+7i\xi_{n}+3)} {-(\xi_{n}-i)^{4}(1+\xi_{n}^{2})^{3}}.
\end{equation}
Therefore
\begin{eqnarray}
\text{ Case \ (7) }
&=&\frac{1}{2}\big(h'(0)\big)^{2}\int_{|\xi'|=1}\int_{-\infty}^{+\infty}
\frac{3(-2\xi_{n}^{4}+3i\xi_{n}^{3}-3\xi_{n}^{2}+7i\xi_{n}+3)} {(\xi_{n}-i)^{4}(1+\xi_{n}^{2})^{3}}
\texttt{d}\xi_{n}\sigma(\xi')\texttt{d}x'\nonumber\\
&=&\frac{1}{2}\big(h'(0)\big)^{2}  \frac{ 2\pi i }{6!}\bigg[
\frac{3(-2\xi_{n}^{4}+3i\xi_{n}^{3}-3\xi_{n}^{2}+7i\xi_{n}+3)}
{(\xi_{n}+i)^{3}}\bigg]^{(6)}\bigg|_{\xi_{n}=i}\Omega_{3}\texttt{d}x'  \nonumber\\
&=&\frac{39}{32}\big(h'(0)\big)^{2}\pi\Omega_{3}\texttt{d}x'.
\end{eqnarray}

\textbf{Case (8)}: \ $r=-1, \ \ell=-2, \ k=0, \ j=0, \ |\alpha|=1$

From (2.20) and the Leibniz rule, we obtain
\begin{eqnarray}
\text{ Case \ (8)}&=&-\int_{|\xi'|=1}\int_{-\infty}^{+\infty}\sum_{|\alpha|=1}
\text{trace}\Big[\partial_{\xi'}^{\alpha}\pi_{\xi_{n}}^{+}\sigma_{-1}(D^{-1})
\partial_{x'}^{\alpha}\partial_{\xi_{n}}\sigma_{-2}(D^{-1})\Big](x_{0})\texttt{d}\xi_{n}\sigma(\xi')\texttt{d}x' \nonumber\\
  &=&\int_{|\xi'|=1}\int_{-\infty}^{+\infty}\sum_{|\alpha|=1}
\text{trace}\Big[\partial_{\xi_{n}}\partial_{\xi'}^{\alpha}\pi_{\xi_{n}}^{+}\sigma_{-1}(D^{-1})
\partial_{x'}^{\alpha}\sigma_{-2}(D^{-1})\Big](x_{0})\texttt{d}\xi_{n}\sigma(\xi')\texttt{d}x'.
\end{eqnarray}
By Lemma 3.2, a simple computation shows
\begin{eqnarray}
\partial_{\xi'}^{\alpha}\sigma_{-1}(D^{-1})(x_{0})\Big|_{|\xi'|=1}
&=&\partial_{\xi_{i}}\sigma_{-1}(D^{-1})(x_{0})\Big|_{|\xi'|=1}  \nonumber\\
&=&\frac{\sqrt{-1}c(\texttt{d}x_{i})}{(1+\xi_{n}^{2})}
-\frac{2\sqrt{-1}\xi_{i}c(\xi) }{(1+\xi_{n}^{2})^{2}}.
\end{eqnarray}
By (3.43) and (3.44), we obtain
\begin{equation}
\pi_{\xi_{n}}^{+}\partial_{\xi'}^{\alpha}\sigma_{-1}(D^{-1})(x_{0})\Big|_{|\xi'|=1}
=\frac{1}{2(\xi_{n}-i)}c(\texttt{d}x_{i})-\xi_{i}\frac{\xi_{n}-2i}{2(\xi_{n}-i)^{2}}c(\xi')
-\xi_{i}\frac{1}{2(\xi_{n}-i)^{2}}c(\texttt{d}x_{n}).
\end{equation}
Then
\begin{equation}
\partial_{\xi_{n}}\pi_{\xi_{n}}^{+}\partial_{\xi'}^{\alpha}\sigma_{-1}(D^{-1})(x_{0})\Big|_{|\xi'|=1}
=\frac{-1}{2(\xi_{n}-i)^{2}}c(\texttt{d}x_{i})-\xi_{i}\frac{3i-\xi_{n}}{2(\xi_{n}-i)^{3}}c(\xi')
+\xi_{i}\frac{1}{(\xi_{n}-i)^{3}}c(\texttt{d}x_{n}).
\end{equation}

Considering for $i<n$,  $\int_{|\xi'|=1}\{\xi_{i_1}\xi_{i_2}\cdots\xi_{i_{2d+1}}\}\sigma(\xi')=0$.
By the relation of the Clifford action and $\texttt{tr}AB=\texttt{tr}BA$, then
\begin{eqnarray}
&&\sum_{ \beta, i, s, \alpha}R^{\partial_{M}}_{ \beta i s \alpha}(x_0)
\texttt{tr}[ c(\texttt{d}x_{i})c(\xi)c(\widetilde{E}_{\beta})c(\widetilde{E}_{s})c(\widetilde{E}_{\alpha})c(\xi)]\nonumber\\
&=&\sum_{ \beta, i, s, \alpha}R^{\partial_{M}}_{ \beta i s \alpha}(x_0)
\texttt{tr}[ c(\xi)c(\widetilde{E}_{i})c(\xi)c(\widetilde{E}_{\beta})c(\widetilde{E}_{s})c(\widetilde{E}_{\alpha})]\nonumber\\
&=&\sum_{ \beta, i, s, \alpha}R^{\partial_{M}}_{ \beta i s \alpha}(x_0)
\texttt{tr}[\big(-c(\widetilde{E}_{i}) c(\xi)-2\xi_{i}\big)c(\xi)c(\widetilde{E}_{\beta})c(\widetilde{E}_{s})c(\widetilde{E}_{\alpha})]\nonumber\\
&=&\sum_{ \beta, i, s, \alpha}R^{\partial_{M}}_{ \beta i s \alpha}(x_0)
\texttt{tr}[\big(1+\xi_{n}^{2}\big)c(\widetilde{E}_{i})c(\widetilde{E}_{\beta})c(\widetilde{E}_{s})c(\widetilde{E}_{\alpha})]
-\sum_{ \beta, i, s, \alpha}R^{\partial_{M}}_{ \beta i s \alpha}(x_0)
\texttt{tr}[2\xi_{i}c(\xi)c(\widetilde{E}_{\beta})c(\widetilde{E}_{s})c(\widetilde{E}_{\alpha})]\nonumber\\
&=&\big(1+\xi_{n}^{2}\big)\sum_{ \beta, i, s, \alpha}R^{\partial_{M}}_{ \beta i s \alpha}(x_0)
\Big(-\delta_{i}^{s}\delta_{\alpha}^{\beta}+\delta_{\alpha}^{i}\delta_{\beta}^{s}\Big)\texttt{tr}[\texttt{id}]\nonumber\\
&&-2\sum_{ \beta, i, s, \alpha}R^{\partial_{M}}_{ \beta i s \alpha}(x_0)\xi_{i}\xi_{\gamma}
\texttt{tr}[c(\widetilde{E}_{\gamma})c(\widetilde{E}_{\beta})c(\widetilde{E}_{s})c(\widetilde{E}_{\alpha})]\texttt{tr}[\texttt{id}]\nonumber\\
&=&\big(1+\xi_{n}^{2}\big)\Big(-\sum_{ \beta, i}R^{\partial_{M}}_{ \beta i i \beta}(x_0)+
\sum_{ \beta, i}R^{\partial_{M}}_{ \beta i \beta i }(x_0)\Big)\texttt{tr}[\texttt{id}]
-2\sum_{ \beta, i, s, \alpha}R^{\partial_{M}}_{ \beta i s \alpha}(x_0)\xi_{i}\xi_{\gamma}
\Big(-\delta_{\gamma}^{s}\delta_{\alpha}^{\beta}+\delta_{\alpha}^{\gamma}\delta_{\beta}^{s}\Big)\texttt{tr}[\texttt{id}]\nonumber\\
&=&8\big(1+\xi_{n}^{2}\big)s_{\partial_{M}}+16\sum_{i, s, \alpha}R^{\partial_{M}}_{ \alpha i s \alpha}(x_0)\xi_{i}\xi_{s}.
\end{eqnarray}
Similarly, we have
\begin{eqnarray}
&& \sum_{i,l,t<n}\xi_{l}\Big(R^{\partial_{M}}_{tilj}(x_0)+R^{\partial_{M}}_{tjli}(x_0)\Big)
\texttt{tr}[ c(\texttt{d}x_{i})c(\xi)c(\texttt{d}x_{j}) c(\widetilde{E}_{t})]
=16\sum_{i,l,t<n}R^{\partial_{M}}_{tilt}(x_0)\xi_{l}\xi_{i},\\
&&\sum_{ \beta, i, s, \alpha}R^{\partial_{M}}_{ \beta i s \alpha}(x_0)\xi_{i}
\texttt{tr}[ c(\xi')c(\xi)c(\widetilde{E}_{\beta})c(\widetilde{E}_{s})c(\widetilde{E}_{\alpha})c(\xi)]
=8(\xi_{n}^{2}-1)\sum_{i, s, \alpha}R^{\partial_{M}}_{ \alpha i s \alpha}(x_0)\xi_{i}\xi_{s},\\
&&\sum_{i,\alpha, \beta<n}\Big(R^{\partial_{M}}_{i \alpha j \beta}(x_0)+R^{\partial_{M}}_{i\beta j \alpha}(x_0)\Big)
\xi_{i}\xi_{\alpha}\xi_{\beta}\texttt{tr}[ c(\xi')c(\xi)c(\texttt{d}x_{i}) c(\xi)]\nonumber\\
&&~~~~=8(1-\xi_{n}^{2})\sum_{i,j,\alpha, \beta<n}R^{\partial_{M}}_{i \alpha j \beta}(x_0)\xi_{i}\xi_{j}\xi_{\alpha}\xi_{\beta}.
\end{eqnarray}
From (3.28), (3.56), (3.57), (3.82)-(3.87) and direct computations, we obtain
\begin{eqnarray}
&&\text{trace}\Big[\partial_{\xi'}^{\alpha}\pi_{\xi_{n}}^{+}\sigma_{-1}(D^{-1})
\partial_{x'}^{\alpha}\partial_{\xi_{n}}\sigma_{-1}(D^{-1})\Big](x_{0})\nonumber\\
&=&s_{\partial_{M}} \frac{-1}{2(\xi_{n}-i)(1+\xi_{n}^{2})}
 +\sum_{\alpha,i,l <n}R^{\partial_{M}}_{\alpha il\alpha }(x_0) \xi_{l}\xi_{i}
 \frac{3\xi_{n}^{3}-9i\xi_{n}^{2}+22\xi_{n}^{2}-44i\xi_{n}-21\xi_{n}-22+15i}
 {-6(\xi_{n}-i)^{3}(1+\xi_{n}^{2})^{2}}.
\end{eqnarray}
Then by (3.56) and (3.88) we have
\begin{eqnarray}
\text{ Case \ (8) }
&=&s_{\partial_{M}}
\int_{|\xi'|=1}\int_{-\infty}^{+\infty}\frac{-1}{2(\xi_{n}-i)(1+\xi_{n}^{2})} \texttt{d}\xi_{n}\sigma(\xi')\texttt{d}x'
+\int_{|\xi'|=1}\sum_{\alpha,i,l <n}R^{\partial_{M}}_{\alpha il\alpha }(x_0) \xi_{l}\xi_{i}\sigma(\xi')
\nonumber\\
&&\times
\int_{-\infty}^{+\infty}\frac{3\xi_{n}^{3}-9i\xi_{n}^{2}+22\xi_{n}^{2}-44i\xi_{n}-21\xi_{n}-22+15i}
 {-6(\xi_{n}-i)^{3}(1+\xi_{n}^{2})^{2}}\texttt{d}\xi_{n}\texttt{d}x'\nonumber\\
&=&s_{\partial_{M}}2\pi i\bigg[\frac{-1}
{2(\xi_{n}+i)}\bigg]^{(1)}\bigg|_{\xi_{n}=i}\Omega_{3}\texttt{d}x'
\nonumber\\
&&+\sum_{i,l <n}R^{\partial_{M}}_{\alpha il\alpha }(x_0) \frac{\pi^{2}}{2}\delta_{l}^{i}
\frac{ 2\pi i }{4!}\bigg[\frac{3\xi_{n}^{3}-9i\xi_{n}^{2}+22\xi_{n}^{2}-44i\xi_{n}-21\xi_{n}-22+15i}
{-6(\xi_{n}+i)^{2}}\bigg]^{(4)}\bigg|_{\xi_{n}=i}\texttt{d}x' \nonumber\\
&=&\frac{- i }{4}s_{\partial_{M}}\pi\Omega_{3}\texttt{d}x'-\frac{1 }{4}s_{\partial_{M}}\pi\Omega_{3}
(-\frac{3 }{4}-\frac{11i}{8})\texttt{d}x'\nonumber\\
&=&\Big(\frac{3 }{16}+\frac{3}{32} i\Big)s_{\partial_{M}}\pi\Omega_{3}\texttt{d}x'.
\end{eqnarray}

\textbf{Case (9)}: \ $r=-1, \ \ell=-2, \ k=1, \ j=0, \ |\alpha|=0$

From (2.20) and the Leibniz rule, we obtain
\begin{eqnarray}
\text{ Case \ (9)}&=&-\frac{1}{2}\int_{|\xi'|=1}\int_{-\infty}^{+\infty}\sum_{|\alpha|=1}
\text{trace}\Big[\partial_{\xi_{n}}\pi_{\xi_{n}}^{+}\sigma_{-1}(D^{-1})
\partial_{\xi_{n}}\partial_{x_{n}}\sigma_{-2}(D^{-1})\Big](x_{0})\texttt{d}\xi_{n}\sigma(\xi')\texttt{d}x' \nonumber\\
  &=&\frac{1}{2}\int_{|\xi'|=1}\int_{-\infty}^{+\infty}\sum_{|\alpha|=1}
\text{trace}\Big[\partial_{\xi_{n}}^{2}\pi_{\xi_{n}}^{+}\sigma_{-1}(D^{-1})
\partial_{x_{n}}\sigma_{-2}(D^{-1})\Big](x_{0})\texttt{d}\xi_{n}\sigma(\xi')\texttt{d}x'.
\end{eqnarray}
By (2.2.29) in \cite{Wa3}, we have
\begin{equation}
\partial_{\xi_n}\pi^+_{\xi_n}q_{-1}(x_0)|_{|\xi'|=1}=-\frac{c(\xi')+ic(\texttt{d}x_n)}{2(\xi_n-i)^2}.
\end{equation}
Then
\begin{equation}
\partial_{\xi_n}^{2}\pi^+_{\xi_n}q_{-1}(x_0)|_{|\xi'|=1}=\frac{1}{(\xi_n-i)^{3}}c(\xi')+\frac{i}{2(\xi_n-i)^{3}}c(\texttt{d}x_n).
\end{equation}
Combining (3.29) and (3.92), we obtain
\begin{eqnarray}
&&\text{trace}\Big[\partial_{\xi_{n}}^{2}\pi_{\xi_{n}}^{+}\sigma_{-1}(D^{-1})
\partial_{x_{n}}\sigma_{-2}(D^{-1})\Big](x_{0})\Big|_{|\xi'|=1}\nonumber\\
&=&\big(h'(0)\big)^{2}\frac{4i\xi_{n}^{6}+4\xi_{n}^{5}+12i\xi_{n}^{4}+19\xi_{n}^{3}+13i\xi_{n}^{2}+39\xi_{n}-19i}
 {(\xi_{n}-i)^{3}(1+\xi_{n}^{2})^{4}}\nonumber\\
 &&-h''(0)\frac{2(2i\xi_{n}^{4}+4\xi_{n}^{3}+2i\xi_{n}^{2}+9\xi_{n}-5i)}
 {(\xi_{n}-i)^{3}(1+\xi_{n}^{2})^{3}}.
\end{eqnarray}
Therefore
\begin{eqnarray}
\text{ Case \ (9) }
&=&\frac{1}{2}\big(h'(0)\big)^{2}\int_{|\xi'|=1}\int_{-\infty}^{+\infty}
\frac{4i\xi_{n}^{6}+4\xi_{n}^{5}+12i\xi_{n}^{4}+19\xi_{n}^{3}+13i\xi_{n}^{2}+39\xi_{n}-19i}
 {(\xi_{n}-i)^{3}(1+\xi_{n}^{2})^{4}}\texttt{d}\xi_{n}\sigma(\xi')\texttt{d}x'\nonumber\\
&&-\frac{1}{2}h''(0)\int_{|\xi'|=1}\int_{-\infty}^{+\infty}\frac{2(2i\xi_{n}^{4}+4\xi_{n}^{3}+2i\xi_{n}^{2}+9\xi_{n}-5i)}
 {(\xi_{n}-i)^{3}(1+\xi_{n}^{2})^{3}}\texttt{d}\xi_{n}\sigma(\xi')\texttt{d}x'\nonumber\\
 &=&\frac{1}{2}\big(h'(0)\big)^{2} \Omega_{3}\int_{\Gamma^{+}}
 \frac{4i\xi_{n}^{6}+4\xi_{n}^{5}+12i\xi_{n}^{4}+19\xi_{n}^{3}+13i\xi_{n}^{2}+39\xi_{n}-19i}
 {(\xi_{n}-i)^{3}(1+\xi_{n}^{2})^{4}}\texttt{d}\xi_{n}\texttt{d}x'\nonumber\\
&&-\frac{1}{2}h''(0)\Omega_{3}\int_{\Gamma^{+}}\frac{2(2i\xi_{n}^{4}+4\xi_{n}^{3}+2i\xi_{n}^{2}+9\xi_{n}-5i)}
 {(\xi_{n}-i)^{3}(1+\xi_{n}^{2})^{3}}\texttt{d}\xi_{n}\texttt{d}x'\nonumber\\
&=&\frac{1}{2}\big(h'(0)\big)^{2}  \frac{ 2\pi i }{6!}
\bigg[\frac{4i\xi_{n}^{6}+4\xi_{n}^{5}+12i\xi_{n}^{4}+19\xi_{n}^{3}+13i\xi_{n}^{2}+39\xi_{n}-19i}
{(\xi_{n}+i)^{4}}\bigg]^{(6)}\bigg|_{\xi_{n}=i}\Omega_{3}\texttt{d}x'  \nonumber\\
&&-\frac{1}{2}h''(0) \frac{ 2\pi i }{5!}\bigg[\frac{2(2i\xi_{n}^{4}+4\xi_{n}^{3}+2i\xi_{n}^{2}+9\xi_{n}-5i)}
{(\xi_{n}+i)^{3}}\bigg]^{(5)}\bigg|_{\xi_{n}=i}\Omega_{3}\texttt{d}x'  \nonumber\\
&=&\Big(-\frac{367}{128}\big(h'(0)\big)^{2}+\frac{103}{64}h''(0)\Big)\pi\Omega_{3}\texttt{d}x'.
\end{eqnarray}

\textbf{Case (10)}: \ $r=-2, \ \ell=-1, \ k=0, \ j=1, \ |\alpha|=0$

From (2.20), we have
\begin{equation}
\text{ Case \ (10)}=\frac{i}{6}\int_{|\xi'|=1}\int_{-\infty}^{+\infty}
\text{trace}\Big[\partial_{x_{n}}\pi_{\xi_{n}}^{+}\sigma_{-2}(D^{-1})
\partial_{\xi_{n}}^{2}\sigma_{-1}(D^{-1})\Big](x_{0})\texttt{d}\xi_{n}\sigma(\xi')\texttt{d}x'.
\end{equation}
By the Leibniz rule, trace property and "++" and "-~-" vanishing
after the integration over $\xi_n$ in \cite{FGLS}, then
\begin{eqnarray}
&&\int^{+\infty}_{-\infty}{\rm trace}
\Big[\partial_{x_{n}}\pi_{\xi_{n}}^{+}\sigma_{-2}(D^{-1})
\partial_{\xi_{n}}^{2}\sigma_{-1}(D^{-1})\Big]\texttt{d}\xi_n \nonumber\\
&=& \int^{+\infty}_{-\infty}{\rm trace}
\Big[\partial_{x_{n}}\sigma_{-2}(D^{-1})\partial_{\xi_{n}}^{2}\sigma_{-1}(D^{-1})\Big]\texttt{d}\xi_n
-\int^{+\infty}_{-\infty}{\rm trace}
\Big[\partial_{x_{n}}\sigma_{-2}(D^{-1})\partial_{\xi_{n}}^{2}\pi_{\xi_{n}}^{+}\sigma_{-1}(D^{-1})\Big]\texttt{d}\xi_n.\nonumber\\
\end{eqnarray}
Combining these assertions, we see
\begin{equation}
\text{ Case \ (10)}=\text{ Case \ (9)}- \frac{i}{6}\int_{|\xi'|=1}\int^{+\infty}_{-\infty}{\rm trace}
\Big[\partial_{x_{n}}\sigma_{-2}(D^{-1})\partial_{\xi_{n}}^{2}\sigma_{-1}(D^{-1})\Big]\texttt{d}\xi_{n}\sigma(\xi')\texttt{d}x'.
\end{equation}
By Lemma 3.2, a simple computation shows
\begin{equation}
\partial_{\xi_{n}}^{2}\sigma_{-1}(D^{-1})(x_{0})\Big|_{|\xi'|=1}
=\frac{6\xi_{n}^{2}-2}{(1+\xi_{n}^{2})^{3}}\sqrt{-1}c(\xi')+\frac{2\xi_{n}^{3}
-6\xi_{n}}{(1+\xi_{n}^{2})^{3}}\sqrt{-1}c(\texttt{d}x_{n}).
\end{equation}
Combining (3.29) and (3.98), we obtain
\begin{equation}
\text{trace}\Big[\partial_{x_{n}}\sigma_{-2}(D^{-1})\partial_{\xi_{n}}^{2}\sigma_{-1}(D^{-1})\Big](x_{0})
=(h'(0))^{2}\frac{8i\xi_{n}^{5}+8i\xi_{n}^{3}+36i\xi_{n}} {(1+\xi_{n}^{2})^{5}}
+h''(0)\frac{8i\xi_{n}^{5}+24i\xi_{n}^{3}+24i\xi_{n}} {(1+\xi_{n}^{2})^{5}}.
\end{equation}
We note that
\begin{equation}
\int_{-\infty}^{+\infty}\frac{8i\xi_{n}^{5}+8i\xi_{n}^{3}+36i\xi_{n}} {(1+\xi_{n}^{2})^{5}}\texttt{d}\xi_{n}
=\frac{ 2\pi i }{4!}\bigg[\frac{8i\xi_{n}^{5}+8i\xi_{n}^{3}+36i\xi_{n}}{(\xi_{n}+i)^{5}}\bigg]^{(4)}\Big|_{\xi_{n}=i}=0,
\end{equation}
and
\begin{equation}
\int_{-\infty}^{+\infty}\frac{8i\xi_{n}^{5}+24i\xi_{n}^{3}+24i\xi_{n}} {(1+\xi_{n}^{2})^{5}}\texttt{d}\xi_{n}
=0.
\end{equation}
Therefore
\begin{equation}
\text{ Case \ (10)}=\Big(-\frac{367}{128}\big(h'(0)\big)^{2}+\frac{103}{64}h''(0)\Big)\pi\Omega_{3}\texttt{d}x'.
\end{equation}

\textbf{Case (11)}: \ $r=-2, \ \ell=-1, \ k=0, \ j=0, \ |\alpha|=1$

From (2.20), we have
\begin{equation}
\text{ Case \ (11)}=-\int_{|\xi'|=1}\int_{-\infty}^{+\infty}\sum_{|\alpha|=1}
\text{trace}\Big[\partial_{\xi'}^{\alpha}\pi_{\xi_{n}}^{+}\sigma_{-2}(D^{-1})
\partial_{x'}^{\alpha}\partial_{\xi_{n}}\sigma_{-1}(D^{-1})\Big](x_{0})\texttt{d}\xi_{n}\sigma(\xi')\texttt{d}x' .
\end{equation}
By Lemma 3.3, for $i<n$, we have
\begin{equation}
\partial_{x_i}\sigma_{-1}(D^{-1})(x_0)=\partial_{x_i}\left(\frac{\sqrt{-1}c(\xi)}{|\xi|^2}\right)(x_0)=
\frac{\sqrt{-1}\partial_{x_i}[c(\xi)](x_0)}{|\xi|^2}
-\frac{\sqrt{-1}c(\xi)\partial_{x_i}(|\xi|^2)(x_0)}{|\xi|^4}=0.
\end{equation}
So Case (11) vanishes.

\textbf{Case (12)}: \ $r=-2, \ \ell=-1, \ k=1, \ j=0, \ |\alpha|=0$

From (2.20) and the Leibniz rule , we have
\begin{eqnarray}
\text{ Case \ (12)}&=&-\frac{1}{2}\int_{|\xi'|=1}\int_{-\infty}^{+\infty}
\text{trace}\Big[\partial_{x_{n}}\pi_{\xi_{n}}^{+}\sigma_{-2}(D^{-1})
\partial_{\xi_{n}}\partial_{x_{n}}\sigma_{-1}(D^{-1})\Big](x_{0})\texttt{d}\xi_{n}\sigma(\xi')\texttt{d}x'\nonumber\\
&=&\frac{1}{2}\int_{|\xi'|=1}\int_{-\infty}^{+\infty}
\text{trace}\Big[\pi_{\xi_{n}}^{+}\sigma_{-2}(D^{-1})
\partial_{\xi_{n}}^{2}\partial_{x_{n}}\sigma_{-1}(D^{-1})\Big](x_{0})\texttt{d}\xi_{n}\sigma(\xi')\texttt{d}x'.
\end{eqnarray}
By the Leibniz rule, trace property and "++" and "-~-" vanishing
after the integration over $\xi_n$ in \cite{FGLS}, then
\begin{eqnarray}
&&\int^{+\infty}_{-\infty}{\rm trace}
\Big[\pi_{\xi_{n}}^{+}\sigma_{-2}(D^{-1})
\partial_{\xi_{n}}^{2}\partial_{x_{n}}\sigma_{-1}(D^{-1})\Big]\texttt{d}\xi_n \nonumber\\
&=& \int^{+\infty}_{-\infty}{\rm trace}
\Big[\sigma_{-2}(D^{-1})\partial_{\xi_{n}}^{2}\partial_{x_{n}}\sigma_{-1}(D^{-1})\Big]\texttt{d}\xi_n
-\int^{+\infty}_{-\infty}{\rm trace}
\Big[\sigma_{-2}(D^{-1})\partial_{\xi_{n}}^{2}\partial_{x_{n}}\pi_{\xi_{n}}^{+}\sigma_{-1}(D^{-1})\Big]\texttt{d}\xi_n.\nonumber\\
\end{eqnarray}
Combining these assertions, we see
\begin{equation}
\text{ Case \ (12)}=\text{ Case \ (7)}+\frac{1}{2}\int_{|\xi'|=1}\int_{-\infty}^{+\infty}
\text{trace}\Big[\sigma_{-2}(D^{-1})\partial_{\xi_{n}}^{2}\partial_{x_{n}}\sigma_{-1}(D^{-1})\Big](x_{0})\texttt{d}\xi_{n}\sigma(\xi')\texttt{d}x'.
\end{equation}
From (3.42) and direct computations, we obtain
\begin{eqnarray}
\partial_{\xi_{n}}^{2}\partial_{x_{n}}\sigma_{-1}(D^{-1})(x_0)|_{|\xi'|=1}
&=&\frac{6i\xi_n^{2}-2i}{(1+\xi_n^{2})^{3}}\partial_{x_n}[c(\xi')](x_0)
+\sqrt{-1}h'(0)\Big[\frac{4(1-5\xi_n^{2})}{(1+\xi_n^{2})^{4}}c(\xi')\nonumber\\
&&-\frac{12\xi_n(\xi_n^{2}-1)}{(1+\xi_n^{2})^{4}}c(\texttt{d}x_n)\Big].
\end{eqnarray}
 From Lemma 3.8, combining (3.77)  and (3.108), we obtain
\begin{equation}
\text{trace}\Big[\sigma_{-2}(D^{-1})\partial_{\xi_{n}}^{2}\partial_{x_{n}}\sigma_{-1}(D^{-1})\Big](x_{0})
=(h'(0))^{2}\frac{30i\xi_{n}} {(1+\xi_{n}^{2})^{4}}.
\end{equation}
We note that
\begin{equation}
\int_{-\infty}^{+\infty}\frac{30i\xi_{n}} {(1+\xi_{n}^{2})^{4}}\texttt{d}\xi_{n}
=\frac{ 2\pi i }{3!}\bigg[\frac{30i\xi_{n}}{(\xi_{n}+i)^{4}}\bigg]^{(3)}\Big|_{\xi_{n}=i}=0.
\end{equation}
Therefore
\begin{equation}
\text{ Case \ (12)}=\frac{39}{32}\big(h'(0)\big)^{2}\pi\Omega_{3}\texttt{d}x'.
\end{equation}

\textbf{Case (13)}: \ $r=-2, \ \ell=-2, \ k=0, \ j=0, \ |\alpha|=0$

From (2.20) and the Leibniz rule , we have
\begin{eqnarray}
\text{ Case \ (13)}&=&-i\int_{|\xi'|=1}\int_{-\infty}^{+\infty}
\text{trace}\Big[\pi_{\xi_{n}}^{+}\sigma_{-2}(D^{-1})
\partial_{\xi_{n}}\sigma_{-2}(D^{-1})\Big](x_{0})\texttt{d}\xi_{n}\sigma(\xi')\texttt{d}x'\nonumber\\
&=&i\int_{|\xi'|=1}\int_{-\infty}^{+\infty}
\text{trace}\Big[\partial_{\xi_{n}}\pi_{\xi_{n}}^{+}\sigma_{-2}(D^{-1})
\sigma_{-2}(D^{-1})\Big](x_{0})\texttt{d}\xi_{n}\sigma(\xi')\texttt{d}x'.
\end{eqnarray}
By \textbf{Case b} in \cite{Wa3} and Lemma 3.5, we obtain
\begin{equation}
\pi^+_{\xi_n}\sigma_{-2}(D^{-1})(x_0)|_{|\xi'|=1}:=B_1-B_2,
\end{equation}
where
\begin{eqnarray}
B_1&=&h'(0)\frac{2+i\xi_n}{4(\xi_n-i)^{2}}c(\xi')c(\texttt{d}x_n)c(\xi')
+h'(0)\frac{-i}{2(\xi_n-i)^{2}}c(\xi')  +h'(0)\frac{-i\xi_n}{4(\xi_n-i)^{2}}c(\texttt{d}x_n)\nonumber\\
&&+\frac{i}{4(\xi_n-i)^{2}}\partial_{x_n}[c(\xi')](x_0)+\frac{-(2+i\xi_n)}{4(\xi_n-i)^{2}} c(\xi') c(\texttt{d}x_n)\partial_{x_n}[c(\xi')](x_0),
\end{eqnarray}
and
\begin{equation}
B_2=\frac{h'(0)}{2}\left[\frac{c(\texttt{d}x_n)}{4i(\xi_n-i)}+\frac{c(\texttt{d}x_n)-ic(\xi')}{8(\xi_n-i)^2}
+\frac{3\xi_n-7i}{8(\xi_n-i)^3}[ic(\xi')-c(\texttt{d}x_n)]\right].
\end{equation}
Hence in this case,
\begin{eqnarray}
\partial_{\xi_{n}}\big(B_1\big)&=&h'(0)\frac{-i\xi_n-3}{4(\xi_n-i)^{3}}c(\xi')c(\texttt{d}x_n)c(\xi')
+h'(0)\frac{i}{(\xi_n-i)^{3}}c(\xi')  +h'(0)\frac{i\xi_n-1}{4(\xi_n-i)^{3}}c(\texttt{d}x_n)\nonumber\\
&&+\frac{-i}{2(\xi_n-i)^{3}}\partial_{x_n}[c(\xi')](x_0)+\frac{i\xi_n+3}{4(\xi_n-i)^{3}} c(\xi') c(\texttt{d}x_n)\partial_{x_n}[c(\xi')](x_0),
\end{eqnarray}
and
\begin{equation}
\partial_{\xi_{n}}\big(B_2\big)=h'(0)\frac{-2i\xi_n-8}{8(\xi_n-i)^{4}}c(\xi')  +h'(0)\frac{i\xi_n^{2}+4\xi_n-9i}{8(\xi_n-i)^{4}}c(\texttt{d}x_n)
\end{equation}
 From Lemma 3.8, combining (3.77)  and (3.116), we obtain
\begin{equation}
\text{trace}\Big[\partial_{\xi_{n}}\big(B_1\big)\sigma_{-2}(D^{-1})\Big](x_{0})
=\frac{2i\xi_{n}^{5}-10\xi_{n}^{4}+26i\xi_{n}^{3}-10\xi_{n}^{2}+37i\xi_{n}+9}
 {-4(\xi_{n}-i)^{3}(1+\xi_{n}^{2})^{3}}\big(h'(0)\big)^{2}.
\end{equation}
Combining (3.77)  and (3.117), we obtain
\begin{equation}
\text{trace}\Big[\partial_{\xi_{n}}\big(B_2\big)\sigma_{-2}(D^{-1})\Big](x_{0})
=-\frac{2i\xi_{n}^{6}+8\xi_{n}^{5}-24i\xi_{n}^{4}-24\xi_{n}^{3}-35i\xi_{n}^{2}-68\xi_{n}+27i}
 {4(\xi_{n}-i)^{4}(1+\xi_{n}^{2})^{3}}\big(h'(0)\big)^{2}.
\end{equation}
Therefore
\begin{eqnarray}
\text{ Case \ (13) }
&=&i\big(h'(0)\big)^{2}\int_{|\xi'|=1}\int_{-\infty}^{+\infty}
\frac{16\xi_{n}^{5}-60i\xi_{n}^{4}-40\xi_{n}^{3}-82i\xi_{n}^{2}-114\xi_{n}+36i}
 {4(\xi_{n}-i)^{4}(1+\xi_{n}^{2})^{3}}\texttt{d}\xi_{n}\sigma(\xi')\texttt{d}x'\nonumber\\
 &=&i\big(h'(0)\big)^{2} \Omega_{3}\int_{\Gamma^{+}}
 \frac{16\xi_{n}^{5}-60i\xi_{n}^{4}-40\xi_{n}^{3}-82i\xi_{n}^{2}-114\xi_{n}+36i}
 {4(\xi_{n}-i)^{4}(1+\xi_{n}^{2})^{3}}\texttt{d}\xi_{n}\texttt{d}x'\nonumber\\
&=&i\big(h'(0)\big)^{2}  \frac{ 2\pi i }{6!}
\bigg[\frac{16\xi_{n}^{5}-60i\xi_{n}^{4}-40\xi_{n}^{3}-82i\xi_{n}^{2}-114\xi_{n}+36i}
{4(\xi_{n}+i)^{3}}\bigg]^{(6)}\bigg|_{\xi_{n}=i}\Omega_{3}\texttt{d}x'  \nonumber\\
&=&-\frac{821}{256}\big(h'(0)\big)^{2}\pi\Omega_{3}\texttt{d}x'.
\end{eqnarray}

\textbf{Case (14)}: \ $r=-1, \ \ell=-3, \ k=0, \ j=0, \ |\alpha|=0$

From (2.20) and the Leibniz rule , we have
\begin{eqnarray}
\text{ Case \ (14)}&=&-i\int_{|\xi'|=1}\int_{-\infty}^{+\infty}
\text{trace}\Big[\pi_{\xi_{n}}^{+}\sigma_{-1}(D^{-1})
\partial_{\xi_{n}}\sigma_{-3}(D^{-1})\Big](x_{0})\texttt{d}\xi_{n}\sigma(\xi')\texttt{d}x'\nonumber\\
&=&i\int_{|\xi'|=1}\int_{-\infty}^{+\infty}
\text{trace}\Big[\partial_{\xi_{n}}\pi_{\xi_{n}}^{+}\sigma_{-1}(D^{-1})
\sigma_{-3}(D^{-1})\Big](x_{0})\texttt{d}\xi_{n}\sigma(\xi')\texttt{d}x'.
\end{eqnarray}
From (3.31), (3.82)-(3.87), (3.91) and direct computations, we obtain
\begin{eqnarray}
&&\text{trace}\Big[\partial_{\xi_{n}}\pi_{\xi_{n}}^{+}\sigma_{-1}(D^{-1})
\sigma_{-3}(D^{-1})\Big](x_{0})\Big|_{|\xi'|=1}\nonumber\\
&=&-\frac{-8\xi_{n}^{7}+18i\xi_{n}^{6}-12\xi_{n}^{5}+61i\xi_{n}^{4}+26\xi_{n}^{3}+66i\xi_{n}^{2}+78\xi_{n}-25i}
 {2(\xi_{n}-i)^{2}(1+\xi_{n}^{2})^{5}}\big(h'(0)\big)^{2}\nonumber\\
 &&+h''(0)\frac{-2\xi_{n}^{5}+6i\xi_{n}^{4}+2\xi_{n}^{3}+11i\xi_{n}^{2}+14\xi_{n}-5i}
 {(\xi_{n}-i)^{2}(1+\xi_{n}^{2})^{4}}\nonumber\\
 &&+s_{\partial_{M}} \frac{-\xi_{n}^{3}-i\xi_{n}^{2}+\xi_{n}-i}{2(\xi_{n}-i)^{2}(1+\xi_{n}^{2})^{3}}
 +\sum_{\alpha,i,l <n}R^{\partial_{M}}_{\alpha il\alpha }(x_0) \xi_{i}\xi_{l}
 \frac{i\xi_{n}^{2}+12\xi_{n}-11i} {3(\xi_{n}-i)^{2}(1+\xi_{n}^{2})^{3}}.
\end{eqnarray}
Therefore
\begin{eqnarray}
\text{ Case \ (14) }
&=&-i\big(h'(0)\big)^{2}  \frac{ 2\pi i }{6!}
\bigg[\frac{-8\xi_{n}^{7}+18i\xi_{n}^{6}-12\xi_{n}^{5}+61i\xi_{n}^{4}+26\xi_{n}^{3}+66i\xi_{n}^{2}+78\xi_{n}-25i}
{2(\xi_{n}+i)^{5}}\bigg]^{(6)}\bigg|_{\xi_{n}=i}\Omega_{3}\texttt{d}x'  \nonumber\\
&&+i h''(0) \frac{ 2\pi i }{5!}\bigg[\frac{-2\xi_{n}^{5}+6i\xi_{n}^{4}+2\xi_{n}^{3}+11i\xi_{n}^{2}+14\xi_{n}-5i}
{(\xi_{n}+i)^{4}}\bigg]^{(5)}\bigg|_{\xi_{n}=i}\Omega_{3}\texttt{d}x'  \nonumber\\
&&+s_{\partial_{M}}\frac{2\pi i}{4!}\bigg[\frac{-\xi_{n}^{3}-i\xi_{n}^{2}+\xi_{n}-i}
{2(\xi_{n}+i)^{3}}\bigg]^{(4)}\bigg|_{\xi_{n}=i}\Omega_{3}\texttt{d}x'
\nonumber\\
&&+\sum_{i,l <n}R^{\partial_{M}}_{\alpha il\alpha }(x_0) \frac{\pi^{2}}{2}\delta_{l}^{i}
\frac{ 2\pi i }{4!}\bigg[\frac{i\xi_{n}^{2}+12\xi_{n}-11i}
{3(\xi_{n}+i)^{3}}\bigg]^{(4)}\bigg|_{\xi_{n}=i}\texttt{d}x' \nonumber\\
&=&\Big(\frac{239}{64}\big(h'(0)\big)^{2}-\frac{27}{16}h''(0)+\frac{29}{192}s_{\partial_{M}}\Big)\pi\Omega_{3}\texttt{d}x'.
\end{eqnarray}

\textbf{Case (15)}: \ $r=-3, \ \ell=-1, \ k=0, \ j=0, \ |\alpha|=0$

From (2.20) we have
\begin{equation}
\text{ Case \ (15)}=-i\int_{|\xi'|=1}\int_{-\infty}^{+\infty}
\text{trace}\Big[\pi_{\xi_{n}}^{+}\sigma_{-3}(D^{-1})
\partial_{\xi_{n}}\sigma_{-1}(D^{-1})\Big](x_{0})\texttt{d}\xi_{n}\sigma(\xi')\texttt{d}x'.
\end{equation}
By the Leibniz rule, trace property and "++" and "-~-" vanishing
after the integration over $\xi_n$ in \cite{FGLS}, then
\begin{eqnarray}
&&\int^{+\infty}_{-\infty}{\rm trace}
\Big[\pi_{\xi_{n}}^{+}\sigma_{-3}(D^{-1})
\partial_{\xi_{n}}\sigma_{-1}(D^{-1})\Big]\texttt{d}\xi_n \nonumber\\
&=& \int^{+\infty}_{-\infty}{\rm trace}
\Big[\sigma_{-3}(D^{-1})
\partial_{\xi_{n}}\sigma_{-1}(D^{-1})\Big]\texttt{d}\xi_n
-\int^{+\infty}_{-\infty}{\rm trace}
\Big[\sigma_{-3}(D^{-1})
\partial_{\xi_{n}}\pi_{\xi_{n}}^{+}\sigma_{-1}(D^{-1})\Big]\texttt{d}\xi_n.\nonumber\\
\end{eqnarray}
Combining these assertions, we see
\begin{equation}
\text{ Case \ (15)}=\text{ Case \ (14)}-i\int_{|\xi'|=1}\int_{-\infty}^{+\infty}
\text{trace}\Big[\sigma_{-3}(D^{-1})
\partial_{\xi_{n}}\sigma_{-1}(D^{-1})\Big](x_{0})\texttt{d}\xi_{n}\sigma(\xi')\texttt{d}x'.
\end{equation}
By Lemma 3.2, a simple computation shows
\begin{equation}
\partial_{\xi_{n}}\sigma_{-1}(D^{-1})(x_{0})\Big|_{|\xi'|=1}
=\frac{-2\xi_{n}}{(1+\xi_{n}^{2})^{2}}\sqrt{-1}c(\xi')+\frac{1-\xi_{n}^{2}}{(1+\xi_{n}^{2})^{2}}\sqrt{-1}c(\texttt{d}x_{n}).
\end{equation}
Combining (3.31), (3.82)-(3.87)  and (3.127),  we obtain
\begin{eqnarray}
&&\text{trace}\Big[\sigma_{-3}(D^{-1})
\partial_{\xi_{n}}\sigma_{-1}(D^{-1})\Big](x_{0})\Big|_{|\xi'|=1}\nonumber\\
&=&\frac{8\xi_{n}^{5}+24\xi_{n}^{3}+28\xi_{n}^{2}} {(1+\xi_{n}^{2})^{5}}\big(h'(0)\big)^{2}
 +\frac{-4\xi_{n}^{3}-8\xi_{n}} {(1+\xi_{n}^{2})^{4}}h''(0)\nonumber\\
 &&+s_{\partial_{M}} \frac{-\xi_{n}^{5}+4\xi_{n}^{3}+\xi_{n}} {(1+\xi_{n}^{2})^{5}}
 +\sum_{\alpha,i,l <n}R^{\partial_{M}}_{\alpha il\alpha }(x_0) \xi_{i}\xi_{l}
 \frac{20\xi_{n}} {3(1+\xi_{n}^{2})^{4}}.
\end{eqnarray}
We note that
\begin{equation}
\int_{-\infty}^{+\infty}\frac{8\xi_{n}^{5}+24\xi_{n}^{3}+28\xi_{n}^{2}} {(1+\xi_{n}^{2})^{5}}\texttt{d}\xi_{n}
=\frac{ 2\pi i }{4!}\bigg[\frac{8\xi_{n}^{5}+24\xi_{n}^{3}+28\xi_{n}^{2}}{(\xi_{n}+i)^{5}}\bigg]^{(4)}\Big|_{\xi_{n}=i}=0,
\end{equation}
and
\begin{equation}
\int_{-\infty}^{+\infty}\frac{-4\xi_{n}^{3}-8\xi_{n}} {(1+\xi_{n}^{2})^{4}}\texttt{d}\xi_{n}
=\int_{-\infty}^{+\infty}\frac{-\xi_{n}^{5}+4\xi_{n}^{3}+\xi_{n}} {(1+\xi_{n}^{2})^{5}}\texttt{d}\xi_{n}
=\int_{-\infty}^{+\infty}\frac{20\xi_{n}} {3(1+\xi_{n}^{2})^{4}}\texttt{d}\xi_{n}
=0.
\end{equation}
Therefore
\begin{equation}
\text{ Case \ (15)}=\Big(\frac{239}{64}\big(h'(0)\big)^{2}-\frac{27}{16}h''(0)+\frac{29}{192}s_{\partial_{M}}\Big)\pi\Omega_{3}\texttt{d}x'.
\end{equation}

Now $\Phi$  is the sum of the \textbf{case ($1,2,\cdots,15$)} , so
\begin{equation}
\sum_{I=1}^{15} \textbf{case I}=\Big(\frac{399}{256}\big(h'(0)\big)^{2}-\frac{29}{32}h''(0)+\big(\frac{71}{96}
+\frac{3}{32}i\big)s_{\partial_{M}}\Big)\pi\Omega_{3}\texttt{d}x'.
\end{equation}
Hence we conclude that, for  $5$-dimensional compact manifold $M$ with the boundary $\partial M$
 \begin{equation}
Vol_{5}^{(1, 1)}=\frac{1}{16}\int_{\partial_{M}}
\Big(\frac{399}{16}\big(h'(0)\big)^{2}-\frac{29}{2}h''(0)+\big(\frac{71}{6}+\frac{3}{2}i\big)s_{\partial_{M}}\Big)\pi\Omega_{3}{\rm dvol}_{\partial_{M}}.
\end{equation}

Next we recall the Einstein-Hilbert action for manifolds with boundary  (see \cite{Wa3} or \cite{Wa4}),
\begin{equation}
I_{\rm Gr}=\frac{1}{16\pi}\int_Ms{\rm dvol}_M+2\int_{\partial M}K{\rm dvol}_{\partial_M}:=I_{\rm {Gr,i}}+I_{\rm {Gr,b}},
\end{equation}
  where
  \begin{equation}
K=\sum_{1\leq i,j\leq {n-1}}K_{i,j}g_{\partial M}^{i,j};~~K_{i,j}=-\Gamma^n_{i,j},
\end{equation}
and $K_{i,j}$ is the second fundamental form, or extrinsic
curvature. Take the metric in Section 2, then by Lemma A.2 in \cite{Wa3},
$K_{i,j}(x_0)=-\Gamma^n_{i,j}(x_0)=-\frac{1}{2}h'(0),$ when $i=j<n$,
otherwise is zero. For $n=5$, then
  \begin{equation}
K(x_0)=\sum_{i,j}K_{i.j}(x_0)g_{\partial M}^{i,j}(x_0)=\sum_{i=1}^{4}K_{i,i}(x_0)=-2h'(0).
\end{equation}
 So
   \begin{equation}
I_{\rm {Gr,b}}=-4h'(0){\rm Vol}_{\partial M}.
\end{equation}

On the other hand, by Proposition 2.10 in \cite{Wa5}, we have
\begin{lem}
 Let M be a $5$-dimensional compact manifold  with the boundary $\partial M$, then
 \begin{equation}
s_{M}(x_{0})=3\big(h'(0)\big)^{2}-4h''(0)+s_{\partial_{M}}(x_{0}).
\end{equation}
\end{lem}
\begin{proof}
From Proposition 2.10 in \cite{Wa5}, let $B=[0,1)$, $b^{2}=\frac{1}{h(x_{n})}$ and $F=\partial_{M}$,  we obtain $s_{B}=0$,
 $|{\rm grad}_{B}b|^2=(b')^{2}$ and
 \begin{equation}
s_{M}(x_{0})=8b''(x_{0})-12\big(b'(x_{0})\big)^{2}+s_{\partial_{M}}(x_{0}).
\end{equation}
By a simple computation, the lemma as follows.
\end{proof}
Hence from (3.133), (3.136) and (3.138), we obtain
\begin{thm}
 Let M be a $5$-dimensional compact manifold  with the boundary $\partial M$, then
  \begin{equation}
 \widetilde{{\rm Wres}}[(\pi^+D^{-1})^2]=\frac{\pi^{3}}{16}\int_{\partial_{M}}\Big( \frac{225}{64}K^{2}
 +\frac{29}{4}s_{M}\big|_{\partial_{M}}
 + \big(\frac{197}{12}+3i \big)s_{\partial_{M}}\Big) {\rm dvol}_{\partial_{M}}.
\end{equation}
\end{thm}

\section*{ Acknowledgements}
This work was supported by Fok Ying Tong Education Foundation under Grant No. 121003,  NSFC. 11271062 and NCET-13-0721.
The authors also thank the referee for his (or her) careful reading and helpful comments.


\begin{thebibliography}{00}
\bibitem{Gu} V. W. Guillemin.: A new proof of Weyl's formula on the asymptotic distribution of eigenvalues. Adv. Math. 55, no. 2, 131-160, (1985).
\bibitem{Wo} M. Wodzicki.: local invariants of spectral asymmetry. Invent. Math. 75(1), 143-178, (1995).
\bibitem{MA} M. Adler.: On a trace functional for formal pseudo-differential operators and the symplectic
structure of Korteweg-de Vries type equations, Invent. Math. 50, 219-248,(1979).
\bibitem{Co1} A. Connes.: Quantized calculus and applications.  XIth International Congress of Mathematical Physics(Paris,1994),
 Internat Press, Cambridge, MA, 15-36, (1995).
 \bibitem{Co2} A. Connes.: The action functinal in Noncommutative geometry. Comm. Math. Phys. 117, 673-683, (1998).
Rev. Math. Phys.5, 477-532, (1993).

\bibitem{Ka} D. Kastler.: The Dirac Operator and Gravitation. Comm. Math. Phys. 166, 633-643, (1995).
\bibitem{KW} W. Kalau and M. Walze.: Gravity, Noncommutative geometry and the Wodzicki residue. J. Geom. Physics. 16, 327-344,(1995).
\bibitem{Ac} T. Ackermann.: A note on the Wodzicki residue. J. Geom. Phys. 20, 404-406, (1996).
\bibitem{FGLS} B. V. Fedosov, F. Golse, E. Leichtnam, E. Schrohe.: The noncommutative residue for manifolds with boundary.
J. Funct. Anal. 142, 1-31, (1996).
\bibitem{S} E. Schrohe.:  Noncommutative residue, Dixmier's trace, and heat trace expansions on manifolds with boundary. Contemp. Math.
242, 161-186, (1999).

 \bibitem{Wa1} Y. Wang.: Diffential forms and the Wodzicki residue for Manifolds with Boundary. J. Geom. Physics. 56, 731-753, (2006).
 \bibitem{H} S. W. Hawking.: In: S. W. Hawking and W. Israel.(eds.)  General Relativity.
 Cambridge University Press, Cambridge, New York, (1979).
 \bibitem{Wa3} Y. Wang.: Gravity and the Noncommutative Residue for Manifolds with Boundary. Letters in Mathematical Physics. 80, 37-56, (2007).
  \bibitem{RP}  R. Ponge.: Noncommutative Geometry and lower dimensional volumes in Riemannian geometry, Lett. Math. Phys. 83, 1-19 (2008).

\bibitem{Wa4} Y. Wang.: Lower-Dimensional Volumes and Kastler-kalau-Walze Type Theorem for Manifolds with Boundary.
      Commun. Theor. Phys. Vol 54, 38-42, (2010).
\bibitem{WW1} J. Wang and  Y. Wang.: Noncommutative residue and sub-Dirac operators for foliations, J. Math. Phys.  54, 012501 (2013).
\bibitem{WW2} J. Wang and  Y. Wang.: The Kastler-Kalau-Walze type theorem for 6-dimensional manifolds with boundary, arXiv: math.DG/1211.6223.

\bibitem{AT} T. Ackermann and J. Tolksdorf.:  A generalized Lichnerowicz formula, the Wodzicki residue and gravity.
J. Geom. Physics. 19, 143-150,(1996).
\bibitem{PS} F. Pf$\ddot{a}$ffle and C. A. Stephan.: On gravity, torsion and the spectral action principle.
      J. Funct. Anal. 262, 1529-1565,(2012).
\bibitem{PS1} F. Pf$\ddot{a}$ffle and C. A. Stephan.: Chiral Asymmetry and the Spectral Action.
  Commun. Math. Phys. 321, 283-310 (2013).
  \bibitem{WWY} J. Wang, Y. Wang, C. L. Yang.: Dirac operators with torsion and the noncommutative
residue for manifolds with boundary. J. Geom. Phys.  81, 92-111 (2014).
\bibitem{DB} Dobarro, F, $\ddot{U}$nal, B.: Curvature of multiply warped products. J. Geom. Physics. 55, 75-106 (2005).


 \bibitem{BGV} N. Berline; E. Getzler; M. Vergne, Heat kernels and Dirac operators. Springer-
Verlag, Berlin, 1992.
 \bibitem{Y} Y. Yu.: The Index Theorem and The Heat Equation Method, Nankai Tracts in Mathematics-Vol.2, World Scientific Publishing, (2001).
    \bibitem{Zh} W. Zhang.: Local index theorem of Atiyah-Singer for families of Dirac operators. Lect. Notes in Math.
  1369, 351-366. Springer, Verlag, (1989).
 \bibitem{GS} G. Grubb, E. Schrohe.: Trace expansions and the noncommutative residue for manifolds with boundary.
J. Reine Angew. Math. 536, 167-207, (2001).
   \bibitem{Wa5} Y. Wang.: Curvature of multiply warped products with an affine connection.
      To appear in Abstract and Applied Analysis (2014).
\end{thebibliography}
\end{document}